\documentclass{article}
\usepackage{kr}

\usepackage{times}
\usepackage{soul}
\usepackage{url}
\usepackage[hidelinks]{hyperref}
\usepackage[utf8]{inputenc}
\usepackage[small]{caption}
\usepackage{graphicx}
\usepackage{amsmath}
\usepackage{amsthm}
\usepackage{booktabs}
\usepackage{algorithm}
\usepackage{algpseudocode} 
\usepackage{mathtools}
\usepackage{enumerate}
\usepackage{hyperref}
\usepackage{cleveref}
\usepackage[bibliography=common]{apxproof} %proofs to appendix / one biblio ("common") for 

\usepackage{bbding}
\usepackage{amssymb}
\usepackage{xspace}
\usepackage{aliascnt}
\usepackage[xcolor, hyperref, cleveref, notion, quotation, electronic]{knowledge}
\usepackage{mathcommand}
\knowledgeconfigure{quotation, protect quotation={tikzcd}}
\knowledgeconfigure{diagnose line=true, diagnose bar=true}

\definecolor{Dark Ruby Red}{HTML}{7c1b1e}
\definecolor{Dark Blue Sapphire}{HTML}{004452} %darker 16%
\definecolor{Dark Gamboge}{HTML}{be7c00}

\IfKnowledgePaperModeTF{
}{
    \knowledgestyle{intro notion}{color={Dark Ruby Red}, emphasize}
    \knowledgestyle{notion}{color={Dark Blue Sapphire}}
    \hypersetup{
        colorlinks=true,
        breaklinks=true,
        linkcolor={Dark Blue Sapphire}, % Links to sections, pages, etc.
        citecolor={Dark Blue Sapphire}, % Links to bibliography
        filecolor={Dark Blue Sapphire}, % Links to local file
        urlcolor={Dark Blue Sapphire},
    }
    \IfKnowledgeElectronicModeTF{
    }{
        \knowledgeconfigure{anchor point color={Dark Ruby Red}, anchor point shape=corner}
        \knowledgestyle{intro unknown}{color={Dark Gamboge}, emphasize}
        \knowledgestyle{intro unknown cont}{color={Dark Gamboge}, emphasize}
        \knowledgestyle{kl unknown}{color={Dark Gamboge}}
        \knowledgestyle{kl unknown cont}{color={Dark Gamboge}}
    }
}
\knowledge{wrap=\textsf}
  | NL

\knowledge{notion, text={paraNL}, wrap=\textsf}
  | para-NL
  | paraNL

\knowledge{notion, wrap=\textsf}
  | FPT
  | fixed-parameter tractable

\knowledge{text={ExpSpace}, wrap=\textsf}
  | EXPSPACE
  | ExpSpace

\knowledge{text={2ExpSpace}, wrap=\textsf}
  | 2EXPSPACE
  | 2ExpSpace

\knowledge{text={2NExpTime}, wrap=\textsf}
  | 2NEXPTIME
  | 2NExpTime
  | 2NexpTime

\knowledge{text={PSpace}, wrap=\textsf}
  | PSpace
  | PSPACE

\knowledge{text={FP}, wrap=\textsf}
  | FP

\knowledge{text={NP}, wrap=\textsf}
  | NP

  \knowledge{text={co-NP}, wrap=\textsf}
  | coNP

  \knowledge{text={P}, wrap=\textsf}
  | P

  \knowledge{text={\ensuremath{\# \textsf{P}}}, wrap=\textsf}
  | shP

  \knowledge{text={\ensuremath{\textsf{P}^{\# \textsf{P}}}}, wrap=\textsf}
  | PshP
  
\knowledge{text={\ensuremath{\textsf{FP}^{\# \textsf{P}}}}, wrap=\textsf}
  | FPshP

\knowledge{text={\ensuremath{\Sigma^p_2}}, wrap=\textsf}
  | SigmaP2

\knowledge{text={\ensuremath{\Pi^p_2}}, wrap=\textsf}
  | PiP2
  | Pi2
  | pi2
  | Pitwo
  | PiTwo
  | pitwo

\knowledge{text={i.e.}, italic}
  | ie

\knowledge{text={I.e.}, italic}
  | Ie

\knowledge{text={s.t.}, italic}
  | st

\knowledge{text={e.g.}, italic}
  | eg

\knowledge{text={vs.}, italic}
  | vs

\knowledge{text={w.r.t.}, italic}
  | wrt

\knowledge{text={a.k.a.}, italic}
  | aka

\knowledge{text={w.l.o.g.}, italic}
  | wlog

\knowledge{text={cf.}, italic}
  | cf

\knowledge{text={iff}, italic}
  | iff

\knowledge{text={r.e.}, italic}
  | r.e.
  | re
  
\knowledge{text={etc.}, italic}
  | etc

\knowledge{text={respectively}, italic}
 | respectively
 | resp

\knowledge{notion}
 | notion@notice
 | definition@notice

\knowledge{notion}
 | bundled fragment
 | Bundled fragment
 | bundled fragments
 | Bundled fragments

\knowledge{notion}
 | modal operator
 | modal operators
 | Modal Operator
 | Modal Operators

\knowledge{notion}
 | literal
 | Literal
 | literals
 | Literals
 
\knowledge{notion}
 | free variable
 | free variables
 
\knowledge{notion}
 | increasing domain model
 | increasing domain models
 | Increasing domain models
 | Increasing domain model
 
\knowledge{notion}
 | constant domain model
 | constant domain models

\knowledge{notion}
 | world
 | worlds
 | World
 | Worlds
 
\knowledge{notion}
 | domain
 | Domain

\knowledge{notion}
 | local domain

\knowledge{notion}
 | accessibility relation
 
\knowledge{notion}
 | valuation function

\knowledge{notion}
 | relevant

\knowledge{notion}
 | satisfiable
 | satisfiability problem
 | satisfiability

\knowledge{notion}
 | valid

\knowledge{notion}
 | module
 | modules
 | Module
 | Modules

\knowledge{notion}
 | Finite Model property
 | finite model property
 | Finite model property

\knowledge{notion}
 | nested $\forall$ formula
 | nested $\forall$ formulas

\knowledge{notion}
 | nested $\forall$ free

\knowledge{notion}
 | outermost $\exists$ free

\knowledge{notion}
 | outermost variable
 | outermost variables

\knowledge{notion}
 | component
 | components

\knowledge{notion}
 | consistent
 | Consistent

\knowledge{notion}
 | atom
 | Atom
 | atoms
 | Atoms

\knowledge{notion}
 | renamed atom
 | Renamed atom
 | renamed atoms
 | Renamed atoms

\knowledge{notion}
 | module set

\knowledge{notion}
 | $\ourEBBE$ subformula
 | $\ourEBBE$ subformulas

\knowledge{notion}
 | skolem-forest
 | Skolem-forest

\knowledge{notion}
 | expansion of skolem-forest

\knowledge{notion}
 | tableau
 | Tableau
 | tableaux
 | Tableaux

\knowledge{notion}
 | saturated
 | saturated tableau

\knowledge{notion}
 | open
 | open tableau
 | open tableaux

\knowledge{notion}
 | tableau rule
 | Tableau rule
 | tableau rules
 | Tableau rules

\knowledge{notion}
 | clean
 | Clean
 | clean formula
 | clean formulas
 | Clean formulas
 | cleanliness

\knowledge{notion}
 | last node
\newif\ifdraft
 \drafttrue 
\usepackage{pifont}
\newcommand{\cmark}{\ding{51}}
\newcommand{\xmark}{\ding{55}}

\usepackage{soul}
\usepackage{todonotes}
\definecolor{green}{RGB}{0,120,0}
\definecolor{hlyellow}{RGB}{250, 250, 190}
\definecolor{apeditcolor}{RGB}{210,255,210}
\definecolor{varadeditcolor}{RGB}{255,210,210}

\ifdraft
    \newcommand{\sideanantha}[1]{\todo[backgroundcolor=apeditcolor, size=\tiny]{{\bf A:} #1}}
    
    \newcommand{\anantha}[1]{\todo[inline,color=apeditcolor, size=\footnotesize]{{\bf A:} #1}}
    
    \newcommand{\sidevarad}[1]{\todo[backgroundcolor=varadeditcolor, size=\tiny]{{\bf V:} #1}}
    
    \newcommand{\varad}[1]{\todo[inline,color=varadeditcolor, size=\footnotesize]{{\bf V:} #1}}
\else

	\newcommand{\sideanantha}[1]{}
    
    \newcommand{\anantha}[1]{}
    
    \newcommand{\sidevarad}[1]{}
    
    \newcommand{\varad}[1]{}
\fi

\ifdraft
    
\else
    
\fi

\renewcommand{\epsilon}{\varepsilon}
\newcommand{\nat}{\mathbb{N}}
\renewcommand{\implies}{\rightarrow}
\renewcommand{\phi}{\varphi}

\knowledgenewrobustcmd{\Var}{\cmdkl{\textit{vars}}}
\knowledgenewrobustcmd{\Ps}{\cmdkl{\mathcal{P}}}
\knowledgenewrobustcmd{\FV}[1][\phi]{\cmdkl{\textit{FV}(#1)}}
\knowledgenewrobustcmd{\SF}[1][\phi]{\cmdkl{\textit{SF}(#1)}}

\knowledgenewrobustcmd\FOML{\cmdkl{\textup{FOML}}\xspace}
\knowledgenewrobustcmd\FO{\cmdkl{\textup{FO}}\xspace}

\knowledgenewrobustcmd{\AB}{\texttt{AB}\xspace}
\knowledgenewrobustcmd{\EB}{\texttt{EB}\xspace}
\knowledgenewrobustcmd{\BA}{\texttt{BA}\xspace}
\knowledgenewrobustcmd{\BE}{\texttt{BE}\xspace}

\knowledgenewrobustcmd{\ABEB}{\texttt{ABEB}\xspace}
\knowledgenewrobustcmd{\ABBA}{\texttt{ABBA}\xspace}
\knowledgenewrobustcmd{\ABBE}{\texttt{ABBE}\xspace}
\knowledgenewrobustcmd{\EBBA}{\texttt{EBBA}\xspace}
\knowledgenewrobustcmd{\EBBE}{\texttt{EBBE}\xspace}
\knowledgenewrobustcmd{\BABE}{\texttt{BABE}\xspace}

\knowledgenewrobustcmd{\ABEBBA}{\texttt{ABEBBA}\xspace}
\knowledgenewrobustcmd{\ABEBBE}{\texttt{ABEBBE}\xspace}
\knowledgenewrobustcmd{\ABBABE}{\texttt{ABBABE}\xspace}
\knowledgenewrobustcmd{\EBBABE}{\texttt{EBBABE}\xspace}
\knowledgenewrobustcmd{\ABEBBABE}{\texttt{ABEBBABE}\xspace}

\knowledgenewrobustcmd{\ourEBBE}{\cmdkl{\exists\Box+\Box\exists}}

\knowledgenewrobustcmd{\M}{\cmdkl{\mathcal{M}}}
\knowledgenewrobustcmd{\W}{\cmdkl{\mathcal{W}}}
\knowledgenewrobustcmd{\R}{\cmdkl{\mathcal{R}}}
\knowledgenewrobustcmd{\D}{\cmdkl{\mathcal{D}}}
\knowledgenewrobustcmd{\val}{\cmdkl{\rho}}
\knowledgenewrobustcmd{\live}{\cmdkl{\delta}}

\knowledgenewrobustcmd{\phiOne}{\cmdkl{\phi_1}}
\knowledgenewrobustcmd{\psiOne}{\cmdkl{\psi_1}}
\knowledgenewrobustcmd{\psiTwo}{\cmdkl{\psi_2}}
\knowledgenewrobustcmd{\psiThree}{\cmdkl{\psi_3}}
\knowledgenewrobustcmd{\psiFour}{\cmdkl{\psi_4}}
\knowledgenewrobustcmd{\psiFive}{\cmdkl{\psi_5}}
\knowledgenewrobustcmd{\psiSix}{\cmdkl{\psi_6}}
\knowledgenewrobustcmd{\psiSeven}{\cmdkl{\psi_7}}
\knowledgenewrobustcmd{\psiEight}{\cmdkl{\psi_8}}
\knowledgenewrobustcmd{\psiNine}{\cmdkl{\psi_9}}
\knowledgenewrobustcmd{\psiTen}{\cmdkl{\psi_{10}}}
\knowledgenewrobustcmd{\psiEleven}{\cmdkl{\psi_{11}}}
\knowledgenewrobustcmd{\psiTwelve}{\cmdkl{\psi_{12}}}
\knowledgenewrobustcmd{\psiThirteen}{\cmdkl{\psi_{13}}}
\knowledgenewrobustcmd{\psiFourteen}{\cmdkl{\psi_{14}}}
\knowledgenewrobustcmd{\psiFifteen}{\cmdkl{\psi_{15}}}

\knowledgenewrobustcmd{\atom}[1][\phi]{\cmdkl{\Lambda_{#1}}}
\knowledgenewrobustcmd{\SFplus}[1][\phi]{\cmdkl{\textsf{SF}^+(#1)}}
\knowledgenewrobustcmd{\comp}[1][\phi]{\cmdkl{\ensuremath{\textsf{Comp}(#1)}}}
\knowledgenewrobustcmd{\compPlus}[1][\phi]{\cmdkl{\ensuremath{\textsf{Comp}^+(#1)}}}
\knowledgenewrobustcmd{\rename}{\cmdkl{\eta}}

\knowledgenewrobustcmd{\skolemForest}[1][\Gamma]{\cmdkl{\mathcal{F}_{#1}}}
\knowledgenewrobustcmd{\forestVariable}{\cmdkl{~\in~}}

\knowledgenewrobustcmd{\expandedGamma}[1][\Gamma]{\cmdkl{\overline{#1}}}

\knowledgenewrobustcmd{\varsInPhi}[1][\phi]{\cmdkl{\textit{vars}({#1})}}
\knowledgenewrobustcmd{\innerExVar}[1][\phi]{\cmdkl{\textit{inner-}\exists\textit{Var}(#1)}}
\knowledgenewrobustcmd{\outerExVar}[1][\phi]{\cmdkl{\textit{outer-}\exists\textit{Var}(#1)}}
\knowledgenewrobustcmd{\forallVar}[1][\phi]{\cmdkl{\forall\textit{Var}(#1)}}

\knowledgenewrobustcmd{\andSet}[1][\Gamma]{\cmdkl{\hat{#1}}}
\knowledgenewrobustcmd{\forestLabel}{\cmdkl{\Pi}}
\knowledgenewrobustcmd{\lastW}[1][w]{\cmdkl{t_{#1}}}

\knowledgenewrobustcmd{\notInit}{\cmdkl{\#}}
\knowledgenewrobustcmd{\emptyTree}{\cmdkl{\epsilon}}

\knowledgenewrobustcmd{\andRule}{\cmdkl{(\land-\textit{rule})}}
\knowledgenewrobustcmd{\orRule}{\cmdkl{(\lor-\textit{rule})}}
\knowledgenewrobustcmd{\diamondRule}{\cmdkl{(\Diamond-\textit{rule})}}
\knowledgenewrobustcmd{\nestedForallRule}{\cmdkl{(\textit{nested }\forall-\textit{rule})}}
\knowledgenewrobustcmd{\trivialSkolemRule}{\cmdkl{(\textit{trivial-skolem-rule})}}
\knowledgenewrobustcmd{\forallRule}{\cmdkl{(\forall-\textit{rule})}}
\knowledgenewrobustcmd{\existsRule}{\cmdkl{(\exists-\textit{rule})}}
\knowledgenewrobustcmd{\EndRule}{\cmdkl{(\textit{end}-\textit{rule})}}

\knowledgenewrobustcmd{\identity}{\cmdkl{\iota}}

\newtheorem{theorem}{Theorem}
\newtheoremrep{theorem}{Theorem}
\newtheorem{lemma}[theorem]{Lemma}
\newtheoremrep{lemma}[theorem]{Lemma}
\newaliascnt{lemma}{theorem}
\newtheorem{definition}[theorem]{Definition}
\newaliascnt{definiton}{theorem}

\newtheoremrep{corollary}[theorem]{Corollary}
\newtheoremrep{proposition}[theorem]{Proposition}
\newaliascnt{proposition}{theorem}

\newaliascnt{cor}{theorem}
\newtheorem{claim}[theorem]{Claim}
\newaliascnt{claim}{theorem}
\newtheoremrep{claim}[theorem]{Claim}

\crefname{thm}{Theorem}{Theorems}
\crefname{proposition}{Proposition}{Propositions}
\crefname{claim}{Claim}{Claims}
\crefname{lemma}{Lemma}{Lemmas} 
\crefname{definition}{Definition}{Definitions} 
\title{A Decidable Bundled Fragment of First-Order Modal Logic Without Finite Model Property}

\author{%
Varad Joshi$^1$\and
Anantha Padmanabha$^2$\\
\affiliations
$^1$IISER, Pune, India\\
$^2$Department of Computer Science and Engineering, Indian Institute of Technology Madras, Chennai, India\\
\emails
joshi.varad@students.iiserpune.ac.in,
ananthap@cse.iitm.ac.in
}

\begin{document}

\maketitle

\begin{abstract}
  The "satisfiability problem" for First-order Modal Logic (\FOML) is undecidable even for simple fragments like having only unary predicates, two variables "etc" Recently a new way to identify decidable fragments of \FOML has been introduced in \cite{BundledJournal} called the `bundled fragments', where the quantifiers and modalities are restricted to appear together. Since there are many ways to bundle the quantifiers together, some of them lead to (un)decidable fragments. In \cite{BundledJournal} the authors prove a `trichotomy', where they show that every bundled fragment falls into one of the following three categories: $(1)$ Those that satisfy "finite model property" (and hence decidable), $(2)$ Those that are undecidable, and $(3)$ Those that do not satisfy "finite model property" (whose decidability is left open).

  In this paper we collapse the trichotomy into a dichotomy over "increasing domain models" by proving that the one combination that falls into the last category is indeed decidable. 
\end{abstract}

\section{Introduction}
\label{sec-Intro}
First-order Modal Logic (\FOML) extends First-order Logic with "modal operators" ($\Box$ and $\Diamond$). In \FOML we can write formulas of the form $\forall x (\Box P(x) \lor \exists y \Diamond Q(x,y))$. The power to assert modalities along with quantification is useful in various scenarios involving quantified epistemic/temporal settings. For instance, in the ontological/epistemic setting, we can state {\em ``Agent knows that every mammal is an animal''} as $\Box\forall x~(\mathtt{mammal(x)}\rightarrow\mathtt{animal(x)})$. In a legal reasoning scenario, the formula $\forall x~(\mathtt{citizen(x) \land \mathtt{adult(x)}) \rightarrow \Box \mathtt{PayTax(x)}}$ asserts that {\em ``All adult citizens are obligated pay taxes''}. Consequently, \FOML has been used to study the quantified versions of temporal logic \cite{temporal1}\cite{temporal2}\cite{temporal3}, epistemic logic \cite{epistemic1}, planning \cite{planning1} "etc"

Despite its potential ubiquitous applications, \FOML is relatively less studied in the literature compared to its propositional counterpart. One important reason for this is that the logic is computationally unfriendly. Very simply fragments of \FOML already become undecidable. For instance \cite{KripkeUndec} proves that \FOML restricted to just unary predicates is undecidable. Along the same lines, almost all decidable fragments of First-order logic like the two variable fragment, guarded fragments "etc" are undecidable when extended with "modal operators" \cite{FOML2varUndec1,FOML2varUndec2}. 

By now we understand the (un)decidable fragments of First-order Logic quite well with many books dedicated to this topic like \cite{ShelahBook}. On the other hand, Propositional Modal Logic and most of its variants/extensions remain decidable \cite{vardi11997modal}. Compared to how thoroughly First-order Logic and Propositional Modal Logic are studied in the literature, it is a bit surprising that the (un)decidability border for \FOML has not gained the attention of the researchers as much as it deserves. Finding interesting decidable fragments of \FOML remains a challenging pursuit.

\medskip
In the early $2000$s one such decidable fragment of \FOML called the {\em monodic fragment} \cite{Monodic1,Monodic2} was identified. In this restriction, inside the scope of "modal operators" there can be at most one free variable. For instance $\forall x (P(x) \lor \Diamond \exists y Q(x,y))$ is  a monodic formula but $\forall x \exists y(P(x) \lor \Diamond Q(x,y))$ is not a monodic formula.  \cite{Monodic1,Monodic2} prove that most decidable fragments of First-order Logic when extended to \FOML under the monodic restriction remains decidable. The decidability relies on the fact that since every modal formula has at most one free variable, it can be treated similar to unary predicates and can be captured by {\em `$1$-types'}. For over a decade, this was the only known approach to obtain decidable fragments of \FOML.

%In studying decidable syntactic fragments of First-order Logic, we typically study restrictions on quantifier alternation, vocabulary or the number of variables in the formula. This motivates the question of thinking about new ways to restrict the syntax to get decidable fragments of \FOML. 

Recently, in pursuit of the decidable fragments of \FOML a new paradigm has emerged called the {\em `"Bundled fragment"'}. In this restriction, the quantifiers and modalities should always occur together. The formula $\Box\forall x~(\mathtt{mammal(x)}\rightarrow\mathtt{animal(x)})$ is a `bundled formula' because the quantifier $\forall$ is preceded by $\Box$. Similarly we can have other restrictions like $\exists$ operator should be succeeded/preceded by a $\Box$ operator "etc" This idea was initiated in \cite{Wang17} to study epistemic logic of {\em know how, know why "etc",} which uses the $\exists \Box$ restriction ("ie" $\exists$ should always be succeeded by $\Box$). For instance $\exists x \Box (\alpha(x))$ can be interpreted as {\em``The agent knows $x$ which explains why or how $\alpha$ is true''}. \cite{Wang17} proves that restricting formulas where every $\exists$ is succeeded by a $\Box$ gives us a decidable fragment of $\FOML$.

But then there are other natural "bundled fragments" ($4$ in total and their duals: like every $\forall / \exists$ should be succeeded/preceded by $\Box$ operator) and the combinations thereof. For instance a formula that restricts every $\forall$ to be succeeded by a $\Box$ operator can have formulas of the form $\forall x \Box\Big( P(x) \lor \exists y\Diamond Q(x,y)\Big)$ since $\exists y \Diamond \alpha$ would be the dual of $\forall y \Box \alpha$. Combination of one or more bundled fragemtns lead to more expressive and interesting logics. For instance, allowing $\forall\Box + \Box\exists$ can assert formulas like $\Diamond \forall x \forall y \Box( P(x,y))$. Classifying the decidability of all the "bundled fragments" and combinations thereof is a natural question both from the foundations of \FOML perspective and towards obtaining practical tools for \FOML in the long run. Interestingly, the decidability of "bundled fragments" depends on whether we are looking at "increasing domain models" or "constant domain models" \cite{BundledJournal}. On the other hand, "bundled fragments" do not restrict the number of variables, arity of predicates or any other such syntactic parameters.

In \cite{BundledFSTTCS} the authors prove that if we allow both $\exists \Box$ and $\forall \Box$ then we get a {\em decidable fragment of \FOML} over "increasing domain models", but having just $\forall\Box$ is already undecidable over unary predicates on "constant domain models". \cite{BundledJournal} takes this a step further and attempts to classify all possible "bundled fragments" and their combinations into (un)decidable fragments. In particular, they classify them into three categories: $(1)$ Those that satisfy "finite model property" (and are hence decidable); $(2)$ Those that are undecidable; and $(3)$ Those that do not satisfy "finite model property". However, the decidability status of the fragments falling into $(3)$ is left open in \cite{BundledJournal}. Thus, their attempt of classification into (un)decidable fragments falls just short and they manage to prove a {\em `trichotomy'} rather than a {\em dichotomy'}.  \Cref{classification} describes the results from \cite{BundledMFCS} for "increasing domain models"\footnote{We restrict our attention to "Increasing domain models" only. The results in \cite{BundledJournal} also has "constant domain models" and the loosely bundled fragment.}.

\begin{figure}
\begin{center}
\begin{tabular}{|l|l|l|l||l|}
\hline
$\forall\Box$&$\exists\Box$&$\Box\forall$&$\Box\exists$&Remark \\
\hline
\hline

 %\cmark & \xmark & \xmark & \xmark & \\

%\xmark & \cmark & \xmark & \xmark & \\
%\cline{2-5}  \cline{7-7}

%\xmark & \xmark & \cmark & \xmark &    \\
%\cline{2-6} 

%\xmark & \xmark & \xmark & \cmark & \\

%\cline{2-7}
%\noalign{\vskip-2\tabcolsep \vskip-3\arrayrulewidth \vskip\doublerulesep}
%\\
%\cline{2-7} 

\cmark & \cmark & \xmark & \xmark &  Satisfies\\
\cline{1-4} 
\cmark & \xmark & \cmark & \cmark &  "Finite model property"\\
&&&&(Hence decidable)\\
\hline\hline

$\star$ & \cmark & \cmark & $\star$ & \\ 
\cline{1-4}
\cmark& $\star$& \cmark & \cmark& Undecidable\\
\cline{1-4}
\cmark & \cmark & $\star$ & \cmark &\\
\hline\hline

\xmark & \cmark & \xmark & \cmark & Does not satisfy\\
&&&&"Finite model property" \\
&&&&(Decidability is left open\\&&&& in \cite{BundledJournal})\\
\hline

%\cline{2-5}

\end{tabular}
\end{center}
\caption{
Reproducing part of the table from \protect\cite[Figure $1$]{BundledJournal} for the "increasing domain models". Here
$\star$ indicates that the presence/absence of the particular fragment is not relevant. Note that the classification is exhaustive.}

\label{classification}
\end{figure}

Over "increasing domain models" there is one fragment that falls into the third category which allows $\exists\Box$ and $\Box\exists$ formulas. This fragment does not satisfy "finite model property" because we can write a formula of the form $\Diamond \forall x \exists y \Box \alpha$ (since $\Diamond \forall$ is the dual of $\Box \exists$) where the $\forall \exists$ alternation applies to the same world. Using this, we can write a formula that encodes a linear order without maximal element. However, there is no way to have a third variable in the same place to assert `grid like property' or the `transitive property' which is typically needed to encode undecidability.

\paragraph{Our contribution}
The goal of this paper is to collapse the {\em `trichotomy'} of \autoref{classification} to a {\em `dichotomy'}.
In particular, we show that the "bundled fragments" that allows formulas of the form $\exists\Box$ and $\Box\exists$ is decidable. 
%Note that \cite{BundledJournal} proves that this fragment does not have "finite model property". 
%Thus, we will take a logic that does not have finite model property and prove that it is decidable.
Our proof relies on a tableau construction. However, since we cannot hope to have a finite tableau, we define an additional `forest like structure' associated with every node of the tableau. 
%We use this forest to generate the {\em `skolem-witnesses'} wherever we have $\forall \exists$ alternation.

%\medskip
Looking at a broader picture, from the propositional modal logic perspective, there are some decidable logics that do not have "finite model property". However, such logics typically involve `transitive frame restrictions' that force an infinite number of "worlds". However, $\exists\Box + \Box\exists$ fragment of \FOML is different since it can force infinite models even when the number of worlds required to satisfy the formula is finite, by forcing the domain to be infinite (\autoref{thm-EBBE has no FMP}). Looking at $\exists\Box + \Box\exists$ as an extension of First-order Logic, note that most\footnote{There are very few exceptions like the BSR class \cite[Section $6.2$]{ShelahBook} and guarded fix point logics\cite{gradel1999guarded}.} decidability proofs of First-Order Logic and its extensions proceed via "finite model property". Thus $\exists\Box + \Box\exists$ fragment of \FOML is yet an other rare extension of First-order Logic that remains decidable even though it violates "finite model property".

\paragraph{Organization of the paper}
In \autoref{sec-prelims} we define the syntax and semantics of \FOML, and then introduce the "bundled fragments". In \autoref{sec-ebbe} we discuss the $\exists\Box + \Box\exists$ fragment and its unique properties in detail. 
%Further, to explain the core idea, we illustrate the tableau procedure informally on a formula $\phiOne$ (defined in \autoref{sec-ebbe}) that violates "finite model property". 
In \autoref{sec-tableaux} we introduce the tableau procedure for $\exists\Box + \Box\exists$. The next two sections \autoref{sec-soundness} and \autoref{sec-completeness} proves the Soundness and Completeness of the tableau procedure "respectively". We conclude in \autoref{sec-conclusion} with some discussions and future directions.

\section{Preliminary}
\label{sec-prelims}
\AP
The syntax of First-order modal logic ($\intro*\FOML$) is given by extending the First-order logic with ""modal operators"" ($\Box$ and $\Diamond$). Note that we exclude equality, constants, and function symbols from the syntax.

\AP
\begin{definition}[$\FOML$ syntax]
\label{def: FOML syntax}
Given a countable set of predicates $\intro*\Ps$  and  a countable set of variables $\intro*\Var$, the syntax\footnote{Note that we exclude equality, constants, and function symbols from the syntax.} of First-order Modal Logic ($\FOML$) is given by:
$$\alpha ::=  P(x_1,\ldots,x_n)  \mid \neg \alpha \mid \alpha \land \alpha  \mid \exists x \alpha \mid \Box \alpha $$ 
where $P\in \Ps$ has arity $n$ and $x,x_1,\ldots, x_n \in \Var$.
\end{definition}
\AP
The boolean connectives $\lor, \rightarrow$, $\leftrightarrow$, the "modal operator" $\Diamond$, and the quantifier $\forall$ are all defined in the standard way. The set of free variables in a formula, denoted by $\intro*\FV$, is similar to what we have for first-order logic (by ignoring the modalities) and let $\intro*\varsInPhi$ denote the set of all variables that occur in $\phi$. We write $\phi(x)$ to mean that $x$ occurs as a free variable of $\phi$. Also, $\phi[y/x]$ denotes the formula obtained from $\phi$ by replacing every free occurrence of $x$ by $y$. The set of subformulas of a formula $\phi$ denoted by $\intro*\SF$ is defined along the standard lines. A formula of the form $P(x_1,\ldots,x_n)$ or $\neg P(x_1,\ldots,x_n)$ is called a ""literal"".

\AP
\begin{definition}[$\FOML$ structure]
\label{def: FOML structure}
An ""increasing domain model""\footnote{An \FOML structure is called a ""constant domain model"" if $\live(w) = \D$ for all $w\in \W$. Since we only consider "increasing domain models" in this paper, we do not define the "constant domain model" in the main paper.} for $\FOML$ is defined by the tuple $\intro*\M = (\W, \D, \live, \R, \val)$ where 
$\intro*\W$ is a non-empty countable set called ""worlds""; $\intro*\D$ is a non-empty countable set called ""domain""; $\intro*\R\subseteq (\W\times \W)$ is the ""accessibility relation"". The map $\intro*\live:\W\mapsto 2^{\D}$ assigns to each 
$w\in \W$ a \textit{non-empty} ""local domain"" set such that whenever\footnote{The monotonicity condition asserting whenever $(w,v)\in \R$ we have $\live(w)\subseteq \live(v)$, is required for evaluating the free variables present in the formula ("cf" \cite{Cresswell96}). Because of this, the models are called  "increasing domain models".}
$(w,v) \in \R$ we have $\live(w)\subseteq \live(v)$ and 
$\intro*\val: (\W\times \Ps) \mapsto \bigcup\limits_{n}2^{{\D}^n}$ is the ""valuation function"", which specifies the interpretation of predicates at every world over the local domain with appropriate arity.
\end{definition}

For a given model $\M$ we denote $\W^{\M},\R^{\M},$ "etc", to indicate the corresponding components. We simply use $\W,\R,$ "etc", when $\M$ is clear from the context. We write $w\to v$ to mean $(w,v)\in \R$ and $\D_w$ instead of $\live(w)$ for brevity.

To evaluate formulas, we need an assignment function for variables. For a given model $\M$, an assignment  function $\sigma: \Var\mapsto \D$ is ""relevant"" at 
$w \in \W$ if $\sigma(x)\in \delta(w)$ for all $x\in \Var$.  For a given $\sigma$, we use $\sigma_{[x\mapsto d]}$ to denote a new assignment that agrees with $\sigma$ for all variables except $x$ which is mapped to $d$.

\begin{definition}[$\FOML$ semantics]
\label{def: FOML semantics}
Given an $\FOML$ model $\M = (\W, \D, \live, \R, \val)$ and $w \in \W$, and $\sigma$ relevant at $w$, for all $\FOML$ formulas $\alpha$
define $\M,w,\sigma \models \alpha$ inductively as follows:
{\small
$$\begin{array}{lcl}
%\hline
\M, w, \sigma\models P(x_1,\ldots,x_n) &\Leftrightarrow & (\sigma(x_1), \ldots, \sigma(x_n))\in \val(w,P)  \\  
\M, w, \sigma\models \neg\alpha &\Leftrightarrow&   \M, w, \sigma\not\models \alpha \\ 
\M, w, \sigma\models \alpha\land \beta &\Leftrightarrow&  \M, w, \sigma\models \alpha \text{ and }\\ 
&& \M, w, \sigma\models \beta \\ 
\M, w, \sigma\models \exists x \alpha &\Leftrightarrow& \text{there is some $d\in \live(w)$ such }\\&&\text{that }\M, w, \sigma_{[x\mapsto d]}\models \alpha \\
\M, w, \sigma\models \Box  \alpha &\Leftrightarrow&  \text{for every } u\in \W \text{ if $w\rightarrow u$}\\&& \text{then }  \M, u, \sigma\models\alpha \\
%\hline
\end{array}$$
}
\end{definition}
\AP
In the sequel, we will only talk about the "relevant" $\sigma$. Also, while evaluating $\alpha$, it is enough to consider $\sigma$ to be a partial function that gives an assignment for the free variables of $\alpha$. We write $\M,w\models \alpha(a)$ to mean $\M,w,[x\mapsto a] \models \alpha(x)$. 

\AP
A formula $\alpha$ is ""satisfiable"" if there is some "increasing domain model" $\M$ and $w\in \W$ and some assignment $\sigma$ "relevant" at $w$  such that $\M,w,\sigma \models \alpha$.  A formula $\alpha$ is ""valid"" if $\neg \alpha$ is not "satisfiable". A set of formulas $\Gamma$ is "satisfiable" if there is some $\M$ and $w\in \W$ and a $\sigma$ such that for every $\phi\in \Gamma$ we have $\M,w,\sigma \models \phi$. The "satisfiability problem" is the following: given a formula $\phi$, is $\phi$ "satisfiable"?

\AP
For a given formula $\phi$, we denote $|\phi|$ to be the length of the formula which counts the number of symbols occurring in $\phi$. The size of an "increasing domain model" $\M$ is defined as $|\M| = |\W|+|\D|$. 

\AP
Any subset of formulas $L$ of \FOML is called a fragment. We say that $L$ satisfies ""Finite model property"" if there is some computable function $f: \nat \to \nat$ such that for every formula $\phi\in L$,  $\phi$ is "satisfiable" iff $\phi$ is "satisfiable" in a model $\M$ such that $|\M| \le f(|\phi|)$. Intuitively, "finite model property" says that if $\phi$ is "satisfiable" then there is a `small' model $\M$ in which $\phi$ is "satisfiable".

\begin{theorem}[Folklore]
Let $L$ be a fragment of $\FOML$ such that $L$ satisfies "finite model property". Then the "satisfiability problem" for $L$ is decidable.
\end{theorem}

\subsection{Bundled fragments}

%The motivation for `bundling' is to restrict the occurrences of quantifiers using modalities. For instance, allowing only formulas of the form $\forall x \Box \alpha$ is one such bundling. On the other hand, we could also have $\Diamond \exists y\alpha$. Thus, there are many ways to `bundle' the quantifiers and modalities. We call these the `bundled operators\slash modalities' \cite{BundledJournal}. The following syntax defines all possible bundled operators:

\begin{definition}[Bundled-$\FOML$ syntax]
\label{def: bundled-FOML syntax}
The ""bundled fragment"" of $\FOML$ is the set of all formulas constructed by the following syntax:

\noindent
{ \small
$$\alpha ::= P(x_1,\ldots,x_n)  \mid \neg \alpha \mid \alpha \land \alpha \mid \Box \alpha \mid  \forall x \Box  \alpha \mid \exists x \Box  \alpha\mid\Box\forall x\alpha \mid \Box\exists x \alpha $$
}

where $P\in \Ps$ has arity $n$ and $x,x_1,\ldots, x_n \in \Var$.
\end{definition}

Note that $\Box \alpha$ can be defined using any one of the bundled operators where the corresponding quantifier is applied to a variable that does not occur in $\alpha$. However, we retain $\Box\alpha$ in the syntax for technical convenience.  The dual operators of $\forall x \Box  \alpha,~ \exists x \Box  \alpha,~\Box\forall x\alpha$ and $\Box\exists x \alpha$ are given by $\exists x\Diamond \alpha,~\forall x\Diamond \alpha,~\Diamond\exists x \alpha$ and $\Diamond\forall{x} \alpha$ "respectively".

\medskip
By allowing combinations of various bundled operators\footnote{\cite{BundledJournal,BundledMFCS} use the nomenclature like $\ABBE, \ABBABE$ "etc" to name each of these combinations. However since we will be interested in only one of the fragments (\EBBE in the naming convention of \cite{BundledJournal}), we do not use this naming convention and refer to the fragment directly as $\exists\Box + \Box\exists$}, we get logics of different expressive power. For instance if we allow only $\forall\Box$ and $\exists \Box$ (which we denote by the fragment $\forall\Box+\exists\Box$) then we can have formulas like $\forall x \Box (P(x) \land \forall y \Diamond Q(x,y))$ (note that $\forall\Diamond$ is the dual of $\exists\Box$) but we cannot write $\Diamond\forall x P(x)$.

\medskip
Each such combination leads to (un)decidability depending on the expressive power that we get (refer \Cref{classification}). At first glance, the classification in \Cref{classification} seems to have no structure. However, there is a method to the madness\footnote{For more details, refer \cite[Section $2.3$]{BundledJournal}}. If a fragment can express both $\forall x\exists y\Box  \alpha$ and $\forall x \Delta \forall y \Delta \forall z  \alpha$ (where $\Delta$ is either a $\Box$ or is empty). For instance, in $\exists\Box + \Box\forall$ we can express $\forall x\exists y\Box  \alpha$ and $\forall x \Delta \forall y \Delta \forall z  \alpha$  as $\Diamond \forall x \exists y \Box \alpha$ and $\Box \forall x \Box \forall y \Box \forall z \alpha$ "respectively". Similarly, in $\forall\Box+\exists\Box+\Box\exists$ we can express $\forall x\exists y\Box  \alpha$ and $\forall x \Delta \forall y \Delta \forall z  \alpha$ as $\Diamond \exists x \forall y \Box \alpha$ and $\forall x \Box \forall y \Box \forall z \Box \alpha$ "respectively". 

If a fragment cannot express  $\forall x \exists y \Box \alpha$ then the fragment satisfies "finite model property" (and is hence decidable). For instance in $\forall\Box+\Box\forall + \Box\exists$ there is no way to express $\forall x \exists y\Box \alpha$ without having some modal operator between the two quantifiers. Finally if a fragment can express $\forall x \exists y\Box \alpha$ but not $\forall x \Delta \forall y\Delta \forall z \Delta \alpha$ then such fragments do not have "finite model property" but their decidability status is open. $\exists\Box + \Box\exists$ falls under this category since we can write $\Diamond\forall x \exists y \Box \alpha$ but there is no way to have three $\forall$s without having a $\Diamond$ somewhere in between.

\AP
The goal of this paper is to prove the decidability of $\intro*\ourEBBE$ which was left as an open problem in \cite{BundledJournal,BundledMFCS}.

\section{\texorpdfstring{$\exists\Box+\Box\exists$}{pdf} fragment}
\label{sec-ebbe}
The $\exists\Box+\Box\exists$ fragment allows bundles of the form $\exists \Box \alpha$ and $\Box \forall x \alpha$ and their duals $\forall x \Diamond \alpha$ and $\Diamond \exists x \alpha$ "respectively". For instance $\exists x\Box\Big(\Diamond \forall y \exists z \Box P(x,y)\Big)$ is a formula in $\exists\Box+\Box\exists$. Consequently, we can express nested $\forall \exists$ that applies to the same "local domain". It is the $\forall \exists$ nesting that needs to be handled carefully since every witness that is generated will also need a witness. On the other hand, if we want to have nested $\forall \forall \forall$ or $\forall \exists \forall$, there is no way of doing this without having a $\Diamond$ somewhere in between the quantifiers. This is the key property which we will use to prove the decidability of the $\exists\Box+\Box\exists$ fragment.

As proved in \cite{BundledJournal}, $\exists\Box+\Box\exists$ does not satisfy "Finite model property". We now give a formula, adapted from \cite{BundledJournal}, that is satisfiable but only in infinite models. Let $P$ be a binary predicate (We write $Pxy$ to mean $P(x,y)$ to avoid excessive bracketing).

\AP
\begin{tabular}{l r l l l}
$\intro*\phiOne:=$&$  \Diamond\forall x \Big[ $&$~\exists y \Box\Box Pxy~\land \Box\Box \neg Pxx ~\land$\\ 
&&$\Diamond \forall y\Big(~\big[\Diamond Pxy\leftrightarrow \Box Pxy~\big]~\land$\\
&&$\quad\quad\Diamond \forall z\big[ \big(Pxy \land Pyz\big) \implies \big(Pxz\big) \big]\Big)~\Big]$
\end{tabular}

%The formula asserts that the binary predicate $P$ is irreflexive, transitive and serial. The shape of the formula is rather involved since it needs to be expressed in $\EBBE$-fragment.

\begin{theorem}[{\cite[Theorem $4.1$]{BundledJournal}}]
\label{thm-EBBE has no FMP}
The formula $\phiOne$ is satisfiable in a model with infinite $\D$. Moreover, for all increasing domain model $\M$ and $w\in \W$ if $\M,w\models \phiOne$ then there is some $w\to u$ such that $\live(u)$ is infinite.
\end{theorem}
\begin{proofsketch}
The main idea is to interpret $Pxy$ as $y$ is a successor of $x$ and hence the first conjunct asserts that every $x$ has a successor $y$ and the second conjunct asserts that the binary predicate $P$ is irreflexive. The last conjunct asserts that $P$ is transitive, and the one before ensures that the relation $Pxy$ holds uniformly across all branches. Together, the formula encodes a linear order without a maximal element. Detailed proof can be found in \cite[Theorem $4.1$]{BundledJournal}.
\end{proofsketch}

\begin{toappendix}
%\subsection{Illustrating the key ideas for \texorpdfstring{$\phiOne$}{pdf}: A primer to the tableau procedure}
\label{subsec-illustration}

In this section we illustrate the key ideas of our tableau construction for the formula $\phiOne$.
From \autoref{thm-EBBE has no FMP} it is clear that $\phiOne$ is satisfiable only in infinite models. We now explain how to get a finite witness for $\phiOne$ from which we can extract a model for $\phiOne$. To handle the subformulas, we define short notations in \Cref{fig-short-notations}.

\AP
\begin{figure}
\begin{center}
\begin{tabular}{|l l l|}
\hline
$\phiOne$&$:=$&$\Diamond\psiOne$\\
\hline

$\intro*\psiOne$&$:=$&$\forall x~ \psiTwo$\\
\hline 

$\intro*\psiTwo(x)$&$:=$&$\psiThree(x)~\land \psiFour(x) ~\land \psiSeven(x)$\\
\hline 
$\intro*\psiThree(x)$&$:=$&$\exists y~ \psiFive(x,y)$\\ 
\hline
$\intro*\psiFour(x)$&$:=$&$\Box\psiSix(x)$\\ 
\hline
$\intro*\psiFive(x,y)$&$:=$&$\Box\psiTwelve(x,y)$\\ 
\hline
$\intro*\psiSix(x)$&$:=$&$\Box \neg Pxx$\\ 
\hline
$\intro*\psiSeven(x)$&$:=$&$\Diamond \psiEight(x)$\\
\hline
$\intro*\psiEight(x)$&$:=$&$\forall y~\psiNine(x,y)$\\
\hline
$\intro*\psiNine(x,y)$&$:=$&$\psiTen(x,y)~\land~ \psiThirteen(x,y)$\\
\hline
$\intro*\psiTen(x,y)$&$:=$&$\psiEleven(x,y)~\leftrightarrow~ \psiTwelve(x,y)$\\
\hline
$\intro*\psiEleven(x,y)$&$:=$&$\Diamond Pxy$\\
\hline
$\intro*\psiTwelve(x,y)$&$:=$&$\Box Pxy$\\
\hline
$\intro*\psiThirteen(x,y)$&$:=$&$\Diamond \psiFourteen(x,y)$\\
\hline
$\intro*\psiFourteen(x,y)$&$:=$&$\forall z~\psiFifteen(x,y,z)$\\
\hline
$\intro*\psiFifteen(x,y,z)$&$:=$&$ \big(Pxy \land Pyz\big) \implies \big(Pxz\big)$\\
\hline
\end{tabular}
\end{center}
 \caption{Short notations for all subformulas of $\phiOne$. The variables in brackets indicate the free variables of the corresponding subformulas.}
\label{fig-short-notations}
\end{figure}

Suppose we want to build a model that satisfies $\phiOne$. Let us start with a "world" $r$ where $\phiOne$ is to be satisfied. Since "local domain" should be non-empty, we start with $\live(r) = d_0$. We denote this by $(r:\{\phiOne\}, \{d_0\})$ to mean that we want the set of formulas $\{\phiOne\}$ to be satisfied at the world $r$ which has $\{d_0\}$ as its "local domain" (we shall add more later, if needed).

Since the leading operator of $\phiOne$ is a $\Diamond$, we create a new world $u$ such that $r\rightarrow u$ and we have $\psiOne$ to be satisfied at $u$. Since $d_0\in \D_r$ we also need to have $d_0\in \D_u$. Hence we want $(u:\{\psiOne\},\{d_0\})$.

\sloppy Now $\psiOne$ is a "nested $\forall$ formula" (which is formally defined in \Cref{sec-tableaux}). Intuitively, a "nested $\forall$ formula" is of the form $\forall x~ \psi'$ where $\psi'$ contains some $\exists$ within the scope of $\forall x$ and therefore, we have a $\forall\exists$ alternation. Such formulas are typically responsible for forcing the infinite domain. 

\medskip
Note that from the syntax of $\ourEBBE$, inside the scope of a "nested $\forall$ formula", we will only have conjunctions/disjunctions/negations of formulas of the form $\exists z \Box \alpha$ or $\forall x \Diamond \alpha$ or "modules". Moreover, a "nested $\forall$ formula" $\forall x~\psi'$ can occur only inside a larger formula of the form $\Diamond\forall x~\psi'$. In particular, it follows from the syntax of $\ourEBBE$ that we can never have $\Box\forall x \psi'$ where $\forall x~\psi'$ is a "nested $\forall$ formula". Hence, we can ensure that at every "world", we will be trying to satisfy at most one "nested $\forall$ formula".

If we strip off the outermost $\forall$ quantifier from the "nested $\forall$ formula" $\forall x~\alpha$, we get a formula $\alpha$ with $x$ as a free variable. Note that there could be other free variables also in $\alpha$ which are inherited from before. But at the current world, $x$ is special since it is $\forall$ quantified. The assignment for all other free variables would have been fixed in the ancestor worlds. We denote $\alpha(S,x)$ to mean that $\FV[\alpha] = S\cup \{x\}$ where $S$ is the set of variables that have been fixed in the ancestor worlds. An "atom" "wrt" $\alpha$ and $x$ (defined in \Cref{sec-tableaux}) is a set of subformulas of $\alpha(S,x)$ which are to be satisfied in the current world to ensure that $\alpha$ itself is satisfied.
For instance, if $\forall x \alpha$ is of the form $\forall x~\Big(\big(\exists y\Box \beta(x,y,u) \lor \Diamond \gamma(x,u)~\big) \land(\forall z\Diamond \delta(x,z)~\big)\Big)$ (where $u$ is a free variable of $\alpha$ that has been quantified earlier), then there are two possible "atoms" "wrt" $\alpha$ and $x$: $\{ \exists y\Box \beta(x,y,u) , \forall z\Diamond \delta(x,z)\}$ and $\{\Diamond \gamma(x,u), \forall z\Diamond \delta(x,z)\}$. Note that if we want to satisfy a "nested $\forall$ formula" $\forall x \alpha$ then for every "local domain" element $d$ there is some "atom" "wrt" $\alpha$ and $x$ that should be satisfied when $x$ is assigned to $d$.
\bigskip

\sloppy Back to the example, we want to satisfy $(u:\{\psiOne\},\{d_0\})$. Now $\psiOne$  is of the form $\forall x \psiTwo(x)$ which is a "nested $\forall$ formula" and by stripping the $\forall$, we get $\psiTwo(x)$ where $\psiTwo(x)$ is of the form $\psiThree(x)\land \psiFour(x)\land \psiSeven(x)$. Note that in this case, there is only one possible "atom" "wrt" $\psi_2(x)$and $x$ given by $\{\psiThree(x), \psiFour(x), \psiSeven(x)\}$. Thus, for every "local domain" $d$ element, we need to ensure that $\{\psiThree(d), \psiFour(d), \psiSeven(d)\} $ is satisfied at $u$.
Since we currently have $d_0$ as the current "local domain" at $u$ we first want to satisfy satisfy $\{\psiThree(d_0), \psiFour(d_0), \psiSeven(d_0)\}$. 

\sloppy Now, $\psiThree(d_0)$ is of the form $\exists y~\psiFive(d_0,y)$, so it needs a witness ($\psiFour(d_0)$ and $\psiSeven(d_0)$ do not need any witnesses and so we will deal with them later). Let us add a fresh domain element $d_1$ which acts as a witness to $d_0$. We will remember that $d_1$ is the $y$-witness of $d_0$ in the form of a tree $T$, where $d_0$ is the root and $d_1$ is the child of $d_0$ indicating that witness relationship. Now $d_1$ has been added to the "local domain" of $u$ and hence we need to satisfy $\{\psiThree(d_1), \psiFour(d_1), \psiSeven(d_1)\}$ at $u$. 

Again $\psiThree(d_1)$ is of the form $\exists y~\psiFive(d_1,y)$, so we pick a new $y$-witness for $d_1$ and call it $d_2$ and add it to the "local domain" of $u$. We also make $d_2$ as the child of $d_1$ in $T$ to remember the witness. Again, we need to satisfy $\psiTwo(d_2)$ for which we need to satisfy $\{ \psiThree(d_2),\psiFour(d_2),\psiSeven(d_2)\}$ at $u$.
\medskip

\sloppy At this point, we again need to satisfy $\{\psiThree(d_2), \psiFour(d_2), \psiSeven(d_2)\}$. But notice that this is the same as the set of formulas that we need to satisfy for $x\to d_0$ and $x\to d_1$ ("ie" the "atom" "wrt" $\psiTwo$ for $d_0$ is the same as the "atom" "wrt" $\psiTwo$ for  $d_1$ which is again the same as the "atom" "wrt" $\psiTwo$ for $d_2$). Hence, we stop creating a new witness for $d_2$. More generally, we stop generating a witness for $d$ if we find two "local domain" elements $c_1$ and $c_2$ with the same "atom" "wrt" the "nested $\forall$ formula" at hand where both $c_1$ and $c_2$  are the ancestors of $d$ in the tree that we are maintaining. In this case, we stopped at $d_2$ because in the tree $T$ we have $d_0\rightarrow d_1\rightarrow d_2$ where $d_0$ and $d_1$ have the same "atom" "wrt" $\psiTwo$ and $x$. 

%The goal is to try to see how to satisfy $\{\psiThree(x), \psiFour(x), \psiSeven(x)\}$ when $x\to d_0$ and $x\to d_1$ and hope that by looking at that, we can deduce how to satisfy the same set of formulas when $x\to d_2$.

\sloppy The key idea is that now we will not try to satisfy $\{\psiThree(d_2), \psiFour(d_2), \psiSeven(d_2)\}$. Instead, we will only have $(u: \{\psiThree(d_0), \psiFour(d_0), \psiSeven(d_0), \psiThree(d_1), \psiFour(d_1), \psiSeven(d_1)\},$ $\{d_0,d_1,d_2\},~ T)$ where $T$ is a tree with $d_0$ as root, $d_1$ is the child of $d_0$ and $d_2$ is the child of $d_1$ (indicating that $d_1$ is witness for $d_0$ and $d_2$ is the witness for $d_1$). Note that we have also removed $\psiOne$ from the set of formulas that we need to satisfy. The hope is that if we know how to satisfy $\{\psiThree(d_0), \psiFour(d_0), \psiSeven(d_0), \psiThree(d_1), \psiFour(d_1), \psiSeven(d_1)\}$, that information is enough to prove that we can generate new witnesses of $d_2$ and thereafter to satisfy and ensure that $\psiOne$ is satisfied at $u$.

\medskip
\sloppy Now, $\psiThree(d_0)$ and $\psiThree(d_1)$ are of the form $\exists y~ \psiFive(d_0,y)$ and $\exists y~ \psiFive(d_0,y)$  "respectively" and $d_1,d_2$ are the witnesses chosen exactly to satisfy these formulas "respectively".
Hence, $(u: \{\psiFive(d_0,d_1), \psiFour(d_0), \psiSeven(d_0), \psiFive(d_1,d_2), \psiFour(d_1), \psiSeven(d_1)\},$ $ \{d_0,d_1,d_2\},~~ T)$. Expanding the formulas for one step, we get $(u: \{\Box\psiTwelve(d_0,d_1), \Box\psiSix(d_0), \Diamond\psiEight(d_0), \Box\psiTwelve(d_1,d_2), $ $\Box\psiSix(d_1), \Diamond\psiEight(d_1)\},~~ \{d_0,d_1,d_2\},~~ T)$.\\
At this point, all formulas that we have at hand have leading "modal operators". Hence, to satisfy $\Diamond\psiEight(d_0)$ and $\Diamond\psiEight(d_1)$ we need to create a new world $v_0$ and $v_1$  such that $u\rightarrow v_0$ and $u\rightarrow v_1$ where we need to satisfy $\{ \psiTwelve(d_0,d_1), \psiSix(d_0), \psiEight(d_0), \psiTwelve(d_1,d_2), \psiSix(d_1)\}$ with $\{d_0,d_1,d_2\}$ as the "local domain" at $v_0$ and satisfy $\{ \psiTwelve(d_0,d_1), \psiSix(d_0), \psiTwelve(d_1,d_2), \psiSix(d_1), \psiEight(d_2)\}$ at $v_1$. Along standard lines, $\Box$ formulas are inherited in all the successor worlds.

\sloppy At this point, the fact that $d_1$ and $d_2$ were used as a witnesses for $d_0$ and $d_1$ "respectively" at $u$ is not relevant anymore. So it is enough to only remember that $d_0,d_1$ and $d_2$ as the "local domain" of $v_0$ and $v_1$ and forget the tree $T$. Hence we have:
\begin{itemize}
    \item $(v_0: \{\psiTwelve(d_0,d_1), \psiSix(d_0), \psiEight(d_0), \psiTwelve(d_1,d_2), \psiSix(d_1)\},$ $\{d_0,d_1,d_2\},\notInit)$
    \item $(v_1: \{ \psiTwelve(d_0,d_1), \psiSix(d_0), \psiTwelve(d_1,d_2), \psiSix(d_1), \psiEight(d_2)\},$ $\{d_0,d_1,d_2\}, \notInit)$ 
\end{itemize}
where $\notInit$ indicates that there is no tree to be remembered. Retrospectively, we have $(r: \{\phiOne\}, \{d_0\}, \notInit)$ at the start.

\medskip
\sloppy  First let us consider $v_0$. Note that $\psiEight(d_0)$ is of the form $\forall y~ \psiNine(d_0,y)$. However, there are no $\exists$ quantifiers in the immediate scope of $\forall y$. Hence $\psiEight$ is not a "nested $\forall$ formula", and to satisfy $\psiEight(d_0)$, we only have to satisfy $\psiNine(d_0,d_0),\psiNine(d_0,d_1)$ and $\psiNine(d_0,d_2)$. 
%Thus, we need to satisfy $\{ \psiTwelve(d_0d_1), \psiTwelve(d_0,d_1), \psiSix(d_0), \psiSix(d_1), \psiNine(d_0d_0), \psiNine(d_0d_2)\}$ at $v_0$.
Now, $\psiNine(x,y)$ is $\psiTen(x,y) \land \psiThirteen(x,y)$. So to satisfy $\{\psiNine(d_0d_0), \psiNine(d_0d_1), \psiNine(d_0,d_2)\}$ we need to satisfy $\{\psiTen(d_0d_i), \psiThirteen(d_0d_i)\mid 0\le i \le 2\}$. But then, $\psiTen(x,y)$ is equivalent to $\Big(\psiEleven(x,y) \land \psiTwelve(x,y)\Big) \lor \Big(\neg \psiEleven(x,y) \land \neg \psiTwelve(x,y) \Big)$. 
We need to choose one disjunction from each of $\psiTen(d_0,d_0), \psiTen(d_0,d_1)$ and $\psiTen(d_0,d_2)$. So we pick as follows:
\begin{itemize}
    \item $\neg \psiEleven(d_0,d_0), \neg \psiTwelve(d_0,d_0)$ for $\psiTen(d_0,d_0)$
    \item $\psiEleven(d_0,d_1), \psiTwelve(d_0,d_1)$ for $\psiTen(d_0,d_1)$
    \item $\psiEleven(d_0,d_2), \psiTwelve(d_0,d_2)$ for $\psiTen(d_0,d_2)$
\end{itemize} 
Hence we have $(v_0: \{\psiTwelve(d_0,d_1), \psiSix(d_0), \psiTwelve(d_1,d_2),$ $ \psiSix(d_1), \neg \psiEleven(d_0,d_0), \neg \psiTwelve(d_0,d_0), \psiEleven(d_0,d_1), \psiTwelve(d_0,d_1), $ $\psiEleven(d_0,d_2), \psiTwelve(d_0,d_2), \psiThirteen(d_0,d_0), \psiThirteen(d_0,d_1), \psiThirteen(d_0,d_2) \},$ $ \{d_0,d_1,d_2\},\notInit)$.

Expanding each of the formula, we get $(v_0: \{\Box P(d_0d_1), \Box \neg P(d_0d_0),\Box P(d_1d_2), \Box \neg P(d_1d_1), \neg\Diamond P(d_0d_0),$ $ \neg \Box P(d_0d_0),  \Diamond P(d_0d_1), \Box P(d_0d_2), \Diamond P(d_0d_2), \Diamond \psiFourteen(d_0,d_0),$ $ \Diamond \psiFourteen(d_0,d_1), \Diamond \psiFourteen(d_0,d_2)  \},~ \{d_0,d_1,d_2\}, \notInit)$.

Rewriting $\neg \Box P(d_0d_0)$ and $\neg \Diamond P(d_0,d_0)$ as $\Diamond \neg P(d_0d_0)$ and $\Box \neg P(d_0d_0)$ "respectively", we get $(v_0: \{\Box P(d_0d_1),$ $ \Box \neg P(d_0d_0), \Box P(d_1d_2), \Box \neg P(d_1d_1), \Diamond \neg P(d_0d_0), \Diamond  P(d_0d_1),$ $ \Box P(d_0d_2), \Diamond P(d_0d_2),  \Diamond \psiFourteen(d_0,d_0), \Diamond \psiFourteen(d_0,d_1), \Diamond \psiFourteen(d_0,d_2)  \},$ $\{d_0,d_1,d_2\}, \notInit)$.
Now all formulas are modal and we have {\em six} $\Diamond$ formulas and so we create {\em six} successors for $v_0$: 
\begin{itemize}
    \item $(w_{0}:~\{ \neg P(d_0d_0), P(d_0d_1), P(d_1d_2), \neg P(d_1d_1), P(d_0,d_2)\},$ $  \{d_0,d_1,d_2\},~ \notInit)$~~  [created for $\Diamond \neg P(d_0,d_0)$]
    \item $(w_{1}:~\{ \neg P(d_0d_0), P(d_0d_1), P(d_1d_2), \neg P(d_1d_1), P(d_0,d_2)\},$ $  ~~\{d_0,d_1,d_2\},~ \notInit)$~~  [created for $\Diamond  P(d_1,d_2)$]
    \item $(w_{2}:~\{ \neg P(d_0d_0), P(d_0d_1), P(d_1d_2), \neg P(d_1d_1), P(d_0,d_2)\},$ $  ~~\{d_0,d_1,d_2\},~ \notInit)$~~  [created for $\Diamond  P(d_0,d_2)$]
    \item $(w_{3}: \{P(d_0d_1), \neg P(d_0d_0), P(d_1d_2), \neg P(d_1d_1), P(d_0d_2),$ $ \psiFourteen(d_0,d_0) \},  \{d_0,d_1,d_2\}, \notInit)$~  [created for $\Diamond  \psiFourteen(d_0,d_0)$]
    \item $(w_{4}: \{P(d_0d_1), \neg P(d_0d_0), P(d_1d_2), \neg P(d_1d_1), P(d_0d_2),$ $ \psiFourteen(d_0,d_1) \},  \{d_0,d_1,d_2\},\notInit)$~  [created for $\Diamond  \psiFourteen(d_0,d_1)$]
    \item $(w_{5}: \{P(d_0d_1), \neg P(d_0d_0), P(d_1d_2), \neg P(d_1d_1), P(d_0d_2),$ $ \psiFourteen(d_0,d_2) \},  \{d_0,d_1,d_2\},\notInit)$~  [created for $\Diamond  \psiFourteen(d_0,d_2)$]
\end{itemize}

Note that $w_{0}, w_{1}$ and $w_{2}$ contains a consistent set of "literals" ("ie" there is no "literal" and its negation that needs to be satisfied together). So the construction stops here for these "worlds".

\noindent
\sloppy Further, $\psiFourteen(x,y)$ is of the form $\forall z \psiFifteen(x,y,z)$ which is not a "nested $\forall$ formula", so there are no witnesses needed. Hence we only need to replace $z$ with each of $d_0,d_1$ and $d_2$ in each of $w_{3},w_{4}$ and $w_{5}$ appropriately. In particular, at $w_3$ we get $(w_{3}: \{P(d_0d_1), \neg P(d_0d_0), P(d_1d_2), \neg P(d_1d_1), P(d_0d_2),$ $ \psiFifteen(d_0,d_0,d_0),\psiFifteen(d_0,d_0,d_1), \psiFifteen(d_0,d_0,d_2) \},$ $\{d_0,d_1,d_2\}, \notInit)$.\\
\sloppy Writing $\psiFifteen$ as $\neg P(x,y) \lor \neg P(y,z) \lor P(x,z)$ and using the right choice for disjunction in each case, we get $(w_{3}: \{P(d_0d_1), \neg P(d_0d_0), P(d_1d_2), \neg P(d_1d_1), P(d_0d_2) \},$ $ \{d_0,d_1,d_2\}, \notInit)$ which is a consistent set of "literals".
Similarly we end up with $(w_{4}: \{P(d_0d_1), \neg P(d_0d_0),$ $ P(d_1d_2), \neg P(d_1d_1), P(d_0d_2), \neg P(d_1d_0) \}, \{d_0,d_1,d_2\}, \notInit)$ and $(w_{5}: \{P(d_0d_1), \neg P(d_0d_0), P(d_1d_2), \neg P(d_1d_1), P(d_0d_2),$ $ \neg P(d_2d_0), \neg P(d_2d_1) \}, \{d_0,d_1,d_2\}, \notInit)$ which are again consistent set of "literals".

\bigskip
A similar construction for $(v_1: \{ \psiTwelve(d_0,d_1), \psiSix(d_0),$ $ \psiTwelve(d_1,d_2), \psiSix(d_1), \psiEight(d_2)\},\notInit)$ will also lead to having leaves where the corresponding sets of formulas are consistent "literals". This is where the tableau construction ends. 

\bigskip

Now, from this tableau, we obtain a model in the canonical way. For each node of the form $(w: \Gamma, C, F)$, we create a "world" $w$ in the model with $C$ as its local domain and $(c,d)\in \val(w,P)$ iff $P(c,d)\in \Gamma$ (along standard lines, we would need many tableaux nodes to resolve $\land/\lor$ "etc", and hence there could be multiple nodes in the tableau with $w$, in which case we look at $\Gamma$ of the last node in the tableau with "world" label $\Gamma$. The construction is formally defined in \Cref{sec-completeness}). The model $\M_0$ obtained from the above constructed tableau for $\phiOne$ is given in \autoref{fig-tableau-example}. Note that one successor for each of $v_0$ and $v_1$ would have done the job, but the way we created the tableau forces us to have {\em six} successors each because of the number of $\Diamond$ formulas.

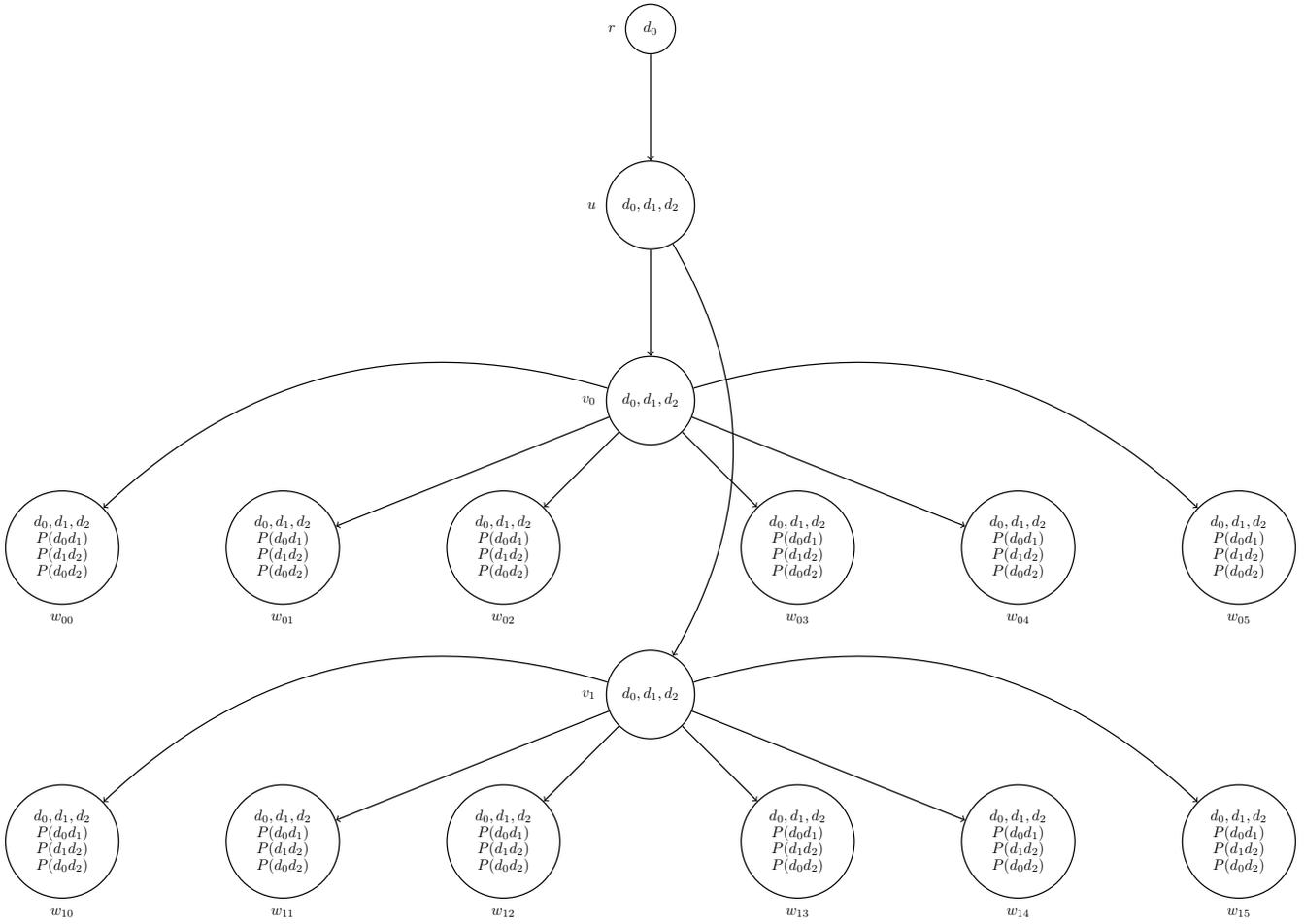
\begin{figure*}
    \centering
    \resizebox{\textwidth}{!}{
    \begin{tikzpicture}[->, node distance=2.5cm, thick, main/.style = {draw, circle, black}]

% Nodes
\node[main] (r) at (0,0) {\begin{tabular}{c} $d_0$ \\  \end{tabular}};
\node[main] (u) [below=of r] {\begin{tabular}{c} $d_0,d_1,d_2$ \\  \end{tabular}};
\node[main] (v0) [below=of u] {\begin{tabular}{c} $d_0,d_1,d_2$ \\ \end{tabular}};
\node[main] (w02) [below left=of v0] {\begin{tabular}{c} $d_0,d_1,d_2$ \\ $P(d_0d_1)$\\ $P(d_1d_2)$\\ $P(d_0d_2)$\end{tabular}};
\node[main] (w03) [below right=of v0] {\begin{tabular}{c} $d_0,d_1,d_2$ \\  $P(d_0d_1)$\\ $P(d_1d_2)$\\ $P(d_0d_2)$\end{tabular}};
\node[main] (w01) [     left=of w02] {\begin{tabular}{c} $d_0,d_1,d_2$ \\ $P(d_0d_1)$\\ $P(d_1d_2)$\\ $P(d_0d_2)$\end{tabular}};
\node[main] (w04) [    right=of w03] {\begin{tabular}{c} $d_0,d_1,d_2$ \\$P(d_0d_1)$\\ $P(d_1d_2)$\\ $P(d_0d_2)$ \end{tabular}};
\node[main] (w00) [     left=of w01] {\begin{tabular}{c} $d_0,d_1,d_2$ \\$P(d_0d_1)$\\ $P(d_1d_2)$\\ $P(d_0d_2)$ \end{tabular}};
\node[main] (w05) [    right=of w04] {\begin{tabular}{c} $d_0,d_1,d_2$ \\ $P(d_0d_1)$\\ $P(d_1d_2)$\\ $P(d_0d_2)$\end{tabular}};

\node[main] (v1) [below right=of w02] {\begin{tabular}{c} $d_0,d_1,d_2$ \\ \end{tabular}};
\node[main] (w12) [below left=of v1] {\begin{tabular}{c} $d_0,d_1,d_2$ \\ $P(d_0d_1)$\\ $P(d_1d_2)$\\ $P(d_0d_2)$\end{tabular}};
\node[main] (w13) [below right=of v1] {\begin{tabular}{c} $d_0,d_1,d_2$ \\ $P(d_0d_1)$\\ $P(d_1d_2)$\\ $P(d_0d_2)$\end{tabular}};
\node[main] (w11) [     left=of w12] {\begin{tabular}{c} $d_0,d_1,d_2$ \\$P(d_0d_1)$\\ $P(d_1d_2)$\\ $P(d_0d_2)$ \end{tabular}};
\node[main] (w14) [    right=of w13] {\begin{tabular}{c} $d_0,d_1,d_2$ \\ $P(d_0d_1)$\\ $P(d_1d_2)$\\ $P(d_0d_2)$\end{tabular}};
\node[main] (w10) [     left=of w11] {\begin{tabular}{c} $d_0,d_1,d_2$ \\$P(d_0d_1)$\\ $P(d_1d_2)$\\ $P(d_0d_2)$ \end{tabular}};
\node[main] (w15) [    right=of w14] {\begin{tabular}{c} $d_0,d_1,d_2$ \\ $P(d_0d_1)$\\ $P(d_1d_2)$\\ $P(d_0d_2)$\end{tabular}};

% Node labels (in red)
\node[ font=\bfseries, left=0.1cm of r] {$r$};
\node[ font=\bfseries, left=0.1cm of u] {$u$};
\node[ font=\bfseries, left=0.1cm of v0] {$v_0$};
\node[ font=\bfseries, left=0.1cm of v1] {$v_1$};

\node[ font=\bfseries, below=0.1cm of w00] {$w_{00}$};
\node[ font=\bfseries, below=0.1cm of w01] {$w_{01}$};
\node[ font=\bfseries, below=0.1cm of w02] {$w_{02}$};
\node[ font=\bfseries, below=0.1cm of w03] {$w_{03}$};
\node[ font=\bfseries, below=0.1cm of w04] {$w_{04}$};
\node[ font=\bfseries, below=0.1cm of w05] {$w_{05}$};

\node[ font=\bfseries, below=0.1cm of w10] {$w_{10}$};
\node[ font=\bfseries, below=0.1cm of w11] {$w_{11}$};
\node[ font=\bfseries, below=0.1cm of w12] {$w_{12}$};
\node[ font=\bfseries, below=0.1cm of w13] {$w_{13}$};
\node[ font=\bfseries, below=0.1cm of w14] {$w_{14}$};
\node[ font=\bfseries, below=0.1cm of w15] {$w_{15}$};

\draw (r) edge   (u);
\draw (u) edge (v0);
\draw (u) edge [bend left] (v1);

\draw (v0) edge [bend right] (w00);
\draw (v0) edge (w01);
\draw (v0) edge (w02);
\draw (v0) edge (w03);
\draw (v0) edge (w04);
\draw (v0) edge [bend left] (w05);

\draw (v1) edge  [bend right] (w10);
\draw (v1) edge (w11);
\draw (v1) edge (w12);
\draw (v1) edge (w13);
\draw (v1) edge (w14);
\draw (v1) edge [bend left] (w15);

\end{tikzpicture}
}
    \caption{Model $\M_0$ constructed from the tableau obtained for $\phiOne$ outlined in \autoref{subsec-illustration}}
    \label{fig-tableau-example}
\end{figure*}

Now, note that $\M_0,r\not\models \phiOne$, this is because at $u$ we have $\M_0,u\not\models \forall x \psiTwo(x)$ which is a "nested $\forall$ formula" and in particular, $\M_0,u,\sigma_{[x\to d_2]} \not\models \exists y~ \psiFive(x)$ where $d_2$ is the leaf in the skolem tree maintained to generate the witness. In general, we will prove that if the original formula $\phiOne$ is not satisfied at the root of the model that we have created, then the reason for this can be traced to some descendant $w$ where some "nested $\forall$ formula" $\forall x \psi$ is false because of some $\exists y \alpha$ inside the scope of $x$ and for some leaf $d$ in the skolem-witness tree maintained at $w$ we have  $\M,w,\sigma_{[x\to d]}\not\models \exists y~\alpha$ and hence $\M_0,w,\sigma_{[x\to d]}\not\models \psi$ (formally stated and proved in \Cref{lemma-leaf-violation}).

Hence, we need to generate a new witness for $d_2$ and name it $d_3$. But then we have $d_0$ which is the first ancestor of $d_2$ which has the same "atom" "wrt" $\psiTwo(x)$. So we ensure at $u$ and all the descendants of $u$, the following `type' information is preserved:
\begin{itemize}
\item $d_3$ has the same "atom" "wrt" $x$ as $d_1$ (because the witness for $d_0$ is $d_1$)
\item $(d_0,d_3)$ {\em mimics} $(d_0,d_2)$ (because $d_2$ has same type as $d_3$ but not an immediate descendant of $d_0$, which is the same relationship between $d_0$ and $d_3$).
\item $(d_1,d_3)$ also {\em mimics} $(d_0,d_2)$ (because $d_2$ has same type as $d_3$ but not an immediate descendant of $d_0$, which is the same relationship between $d_1$ and $d_3$).
\item $(d_2,d_3)$ {\em mimics} $(d_0,d_1)$ (because $d_1$ has same type as $d_3$ is an immediate descendant of $d_0$, which is the same relationship between $d_2$ and $d_3$).
\end{itemize} 
Thus, in the new model $\M_1$ where for every $w\in \W$ which are the descendants of $u$ (including $u$), we have $\live_1(w) = \live_0(w) \cup \{d_3\}$ and for all $w\in \{  w_{00},w_{01},w_{03},w_{04}, w_{05}, w_{06}, w_{10},w_{11},w_{13},w_{14}, w_{15}, w_{16} \}$ we have $\val_1(w,P) = \val_0(w,P) \cup \{ (d_0,d_3), (d_1,d_3), (d_2,d_3)\}$. Note that we added $P(d_0,d_3), P(d_1,d_3)$ and $P(d_2,d_3)$ because of $P(d_0,d_2), P(d_0,d_2)$ and $P(d_0,d_1)$ "respectively". It is easy to see that this construction gives us a well defined model but in general we need to prove that the way we have extended the "valuation function" will be consistent with the earlier model.

Now, we still have $\M_1,r\not\models \phiOne$, again because $d_3$ does not have a witness at $u$. But $d_2$ now has a witness. Thus we extend $\M_1$ to get $\M_2$ where we add a new element $d_4$ as the witness for $d_3$. But then in $\M_2$ we will not have a witness for $d_4$. We thus obtain a sequence of models $\M_0,\M_1,\M_2\ldots$. 

\bigskip
The final model that we have is the limit of this construction $\M$. In that model we argue that $\M,r\models \phiOne$. Otherwise there is some $d_i$ which does not have a witness at $u$, but this is a contradiction since by construction we would have the witness $d_{i+1}$ for $d_i$ in $\M_i$ which will be part of $\M$.

\end{toappendix}

In \Cref{subsec-illustration}, we informally illustrate the main ideas of our construction for the formula $\phiOne$ by showing how to obtain a tableau which is a `finite witness' for the satisfiability of $\phiOne$. We recommend the reader to go through \Cref{subsec-illustration} to get an intution of the tableau construction.

\section{Tableaux rules for \texorpdfstring{$\ourEBBE$}{pdf}}
\label{sec-tableaux}
\AP
%We now generalize the tableau procedure that we discussed for $\phiOne$ in \autoref{subsec-illustration} to an arbitrary $\ourEBBE$ formula.
We will always work with formulas in negation normal form and hence the $\neg$ occurs only in "literals". 
Note that the subformulas of $\ourEBBE$ can have formulas that are not in $\ourEBBE$. For instance $\forall x P(x)$ is not in $\ourEBBE$ but is a subformula of $\Diamond\forall x P(x)$. Hence we need to work with subformulas of $\ourEBBE$. We say a formula $\alpha$ is an ""$\ourEBBE$ subformula"" if there is some $\ourEBBE$ formula $\phi$ such that $\alpha\in \SF$. 
 Let\footnote{We write $\neg\psi$ for brevity, more formally it is the negation normal form equivalent of $\neg \psi$. Same comment applies while defining $\compPlus$ and "atom".} $\intro*\SFplus = \SF \cup \{\neg \psi \mid \psi\in \SF\}$.

\AP
\sloppy A ""module"" is either a "literal" or of the form $\Delta \alpha$ where $\Delta\in \{\Box,\Diamond\}$.  Given a "$\ourEBBE$ subformula" $\phi$, the ""components"" of $\phi$ (denoted by $\intro*\comp$) is defined inductively as follows: if $\phi$ is a "module" then $\comp = \{\phi\}$; $\comp[\phi\odot \psi] = \comp[\phi] \cup \comp[\psi]$ where $\odot\in \{\lor,\land\}$; $\comp[Qx~\phi] = \{Qx~\phi\} \cup \comp[\phi]$ where $Q\in \{\forall,\exists\}$. Intuitively, "module" denotes the set of all `atomic' formulas that need to be evaluated in the current "world" and "component" of a formula collects all subformulas that needs to be evaluated in the current "world". Let $\intro*\compPlus = \comp \cup \{\neg \psi \mid \psi \in \comp\}$.

\bigskip
\AP
We call a "$\ourEBBE$ subformula" $\forall x \alpha$ as a ""nested $\forall$ formula""  if there is some $\beta \in \comp[\forall x \alpha]$ where $\beta$ is of the form $\exists y \gamma$ or $\forall y~\delta$. Note that since we are looking at $\ourEBBE$ fragment, from the syntax it follows that $\gamma$ and $\delta$ have to be of the form $\Box \gamma'$ and $\Diamond \delta'$ "respectively". These "nested $\forall$ formulas" play a critical role since these are the only formulas that can have a nested quantification that apply to the same "local domain". In particular, to satisfy $\forall x \alpha$ where $\alpha$ contains $\exists y\Box \beta$ in it, we need to generate a witness $y$ for every $x$ and each such new witness generated will also need a new witness and so on which can potentially force an infinite domain.

\AP
\sloppy Let $\forall x~\phi$ be a "nested $\forall$ formula" where $S = \FV[\forall x~\phi]$ and $\FV[\phi] = S \cup \{x\}$. We define $\intro*\atom \subseteq \compPlus$ is called an ""atom"" of $\phi$ "wrt" $x$ if the following conditions hold: $(1)~\phi \in \atom$; $(2)$ if $\gamma\land \delta \in \atom$ then $\{\gamma,\delta\} \subseteq \atom$; $(3)$ if $\gamma \lor \delta \in \atom$ then $\gamma\in \atom$ or $\delta\in \atom$; $(4)$  $\psi\in \atom$ iff $\neg \psi \not\in \atom$. Furthermore, we define $\intro*\innerExVar[\atom] = \{ y\mid$ there is some formula of the form $\exists y \psi \in \atom\}$ to be the set of all existential variables that occur in $\atom$.

In a "nested $\forall$ formula" $\forall x~\phi$, even though $\FV[\phi]$ contains variables other than $x$, the variable $x$ is special  since $x$ should be instantiated with every "local domain" element. All other free variables would would have been assigned an interpretation in the ancestors. 
%Given an "atom" $\atom$ of $\phi$ "wrt" $x$, with $\innerExVar[\atom] = \{ y_1,\ldots,y_k\}$, let $z,z_1\ldots z_k$ be a sequence of variables. Define $\overline\atom = \{ \beta[x/z,y_1/z_1,\ldots y_k/z_k] \mid \beta \in \atom\}$. Intuitively, $\overline\atom$ instantiates $\atom$ with $z$ taking the place of $x$ and the witnesses $y_1,\ldots y_k$ are provided by the variables $z_1\ldots z_k$ "respectively".

\AP
On the other hand, if an "$\ourEBBE$ subformula" of the form  $\exists y ~\alpha$ does not appear in the scope of a $\forall$ then such  $\exists$ quantifiers should be treated differently from the $\exists$ appearing in  the nested form $\forall \exists$.  In a formula $\phi$ define $\intro*\outerExVar$ that collects those existentially quantified variables that are not inside the immedidate scope of any $\forall$ quantifier. If $\phi$ is a "module" then $\outerExVar = \emptyset$; $\outerExVar[\phi\odot \psi] = \outerExVar[\phi] \cup \outerExVar[\psi]$ where $\odot\in \{\land,\lor\}$; $\outerExVar[\exists x \phi] = \{x\} \cup  \outerExVar[\phi]$ and $\outerExVar[\forall x \phi] = \emptyset$.

\AP
For a finite set of "$\ourEBBE$ subformulas" $\Gamma$, we  denote $\intro*\andSet$ to be the formula formed by the conjunction of the formulas in $\Gamma$ and use $\FV[\Gamma], \outerExVar[\Gamma]$ "etc", to denote $\FV[\andSet], \outerExVar[\andSet]$ "etc", "respectively". The set $\Gamma$ is ""consistent"" if there is no formula of the form $\psi$ and $\neg \psi$ such that $\{\psi,\neg \psi\} \subseteq \Gamma$. Note that, "atoms" are always "consistent". If there is some "nested $\forall$ formula" $\forall x \psi\in \Gamma$ then we say $\Gamma$ contains a "nested $\forall$ formula"; otherwise $\Gamma$ is ""nested $\forall$ free"".

As we will see, we will be interested in $\Gamma$ where $\Gamma$ is either "nested $\forall$ free" or contains exactly one "nested $\forall$ formula".

\AP
Consider a "nested $\forall$ formula" of the form $\forall x \big(P(x,u) \lor \exists y~\Box P(x,y)\big)$ where $u$ is a "free variable". Suppose the current "local domain" is $\{d_0,d_1,d_2\}$ with $\sigma(u) = d_0$. Depending on which disjunction is true for each $d_i$, we need to decide if they should need a new witness or not. For instance if $\big(P(x,u) \lor \exists y~\Box P(x,y)\big)$ is true for $x\to d_0$ because of $\exists y~\Box P(x,y)$ then we need a new witness for $y$ given by $d_3$. To remember that $d_3$ is the witness for $d_0$, we create a tree, make $d_0$ as the root with $d_3$ as its child. Furthermore if $d_3$ needs a witness, we add $d_4$ and so on. This construction gives us a tree with root at $d_0$. Similarly if $d_1$ needs a witness then we add a new element $d_5$ and make $d_1$ as root with $d_5$ and create more witnesses from there on which again becomes a tree. Thus we get a forest where the root of each tree corresponds to some $d_i$ that is present in the initial "local domain". We stop the construction of a tree if the "atom" corresponding to the leaf (formulas to be satisfied for the leaf variable with respect to the "nested $\forall$ formula" at hand) is the same as the "atom" to be satisfied for at least two other variables along the path from the root to the leaf. Please refer \Cref{subsec-illustration} for an illustrative example. 

\AP
\begin{definition}[Skolem-Forest]
\label{def-skolem-forest}
Let $\Gamma$ be a finite set of "$\ourEBBE$ subformula" that contains exactly one "nested $\forall$ formula" $\forall x~ \psi \in \Gamma$ and let $S = \FV[\Gamma] \cup  \outerExVar[\Gamma]$. Define ""skolem-forest"" "wrt" $(\Gamma,S)$ denoted by $\intro*\skolemForest[\Gamma,S] = (V,E,\intro*\forestLabel)$ as a labeled forest such that the following properties hold:
    \begin{enumerate}
        \item \label{def-item-root} $(V,E)$ is a forest where the vertex set $V$ is a finite set of variables where
            $S \subseteq V$ such that $S = \{ z\mid z$ is a root of some tree$\}$ ("ie" $S$ forms the set of all root nodes in the forest).            
        \item\label{def-item-long-enough-path} $\intro*\forestLabel$ maps every variable $z\in V$ to some "atom" of $\psi$ "wrt" $x$.
        \item For every $z\in V$ if $\forestLabel(z) = \atom[\psi]$ and  $z$ is a leaf node in some tree $T_r$ rooted at the variable $r$ then one of the following is true:
        \begin{enumerate}
            \item $\forestLabel(z)$ does not contain any formula of the form $\exists y~\beta$ (OR)
            \item there exists two distinct variable $z_1,z_2\ne z$ on the path from $r$ to $z$ in $T_r$ such that $\forestLabel(z) = \forestLabel(z_1) = \forestLabel(z_2)$.
        \end{enumerate}  
    \end{enumerate}

Given a "skolem-forest" $\skolemForest[\Gamma,S] = (V,E,\forestLabel)$ define ""expansion of skolem-forest"" $\skolemForest[\Gamma,S]$ denoted by $\intro*\expandedGamma = \Big(\Gamma \cup \Gamma'\Big) \setminus \{\forall x~\psi\}$ where $\Gamma'$ is the smallest set of formulas that satisfies the following conditions:
\begin{itemize}
    \item For every variable $z\in V$ for every formula $\beta\in \forestLabel(z)$ if $\beta$ is not of the form $\exists y~\beta'$ then $\beta[z/x] \in \Gamma'$.
    \item For every variable $z\in V$ if $z$ is not a leaf node in some tree then for every formula $\beta\in \forestLabel(z)$ if $\beta$ is of the form $\exists y~\beta'$ then there is some child node/variable $y_z$ of $z$ such that $\exists y_z~\beta'[x/z,y/y_z] \in \Gamma'$.
\end{itemize}
\end{definition}

Note that the last condition in the definition of the "skolem-forest" ensures that either the leaf variable does not need any witness or we have two predecessors where we have satisfied the same "atom". The "expansion of skolem-forest" adds all formulas that needs to be satisfied along with $\Gamma$. Note that we only add $\exists y_z~\beta'[x/z,y/y_z] \in \Gamma'$ if $z$ is not a leaf node. Intuitively it means we will satisfy the required witnesses only for the non-leaf variables.

\bigskip
\AP
Now we are ready to give the tableau procedure for the $\ourEBBE$ fragment. For technical convenience, we will use the set of variables also as the set of domain elements. In our setting a ""tableau"" $T$ is a tree structure such that each node is represented by a tuple of the form $(w:\Gamma,S,F)$ where $w$ is a symbol or a finite sequence of symbols intended as the \textit{name} of a possible world in the actual model, $\Gamma$ is a finite set of  "$\ourEBBE$ subformulas", $S$ is the set of variables which forms the roots of each tree in  $F$ which is a "skolem-forest". We intend to use the set of variables as the domain in the tableau-induced actual model. The intuitive meaning of the node $(w: \Gamma,S,F)$ is that all the formulas in $\Gamma$ are supposed to be satisfied at $w$ with $S$ as the local domain such that if we have a  $\forall x \exists y$ nesting in some formula in $\Gamma$ then the witness $y$ for every $x$ can be generated by looking at successor of every $x$ in the "skolem-forest" $F$. A ""tableau rule"" specifies how the node in the premise of the rule is transformed to or connected with one or more new nodes given by the conclusion of the rule. Intuitively, a tableau for $\phi$ is a pseudo model which can be transformed into an "increasing domain model" of $\phi$ under some simple consistency conditions.

\AP
  Applying the "tableau rules" generate a tree-like structure and a "tableau" is said to be ""saturated"" if every leaf node contains only "literals". For any formula $\phi$, we refer to a "saturated tableau" of $\phi$ simply as a tableau of $\phi$. Further, a "saturated tableau" is ""open"" if in every node $(w: \Gamma,S,F)$ of the "tableau", $\Gamma$ does not contain both $\beta$ and $\neg \beta$ for any formula $\beta$.

\bigskip
\AP
We call a formula ""clean"" if no variable occurs both bound and free  and
every variable is quantified at most once.  For instance, the formulas
$\exists x \Box Px \lor \forall x \Diamond Qx$ and $Px \land  \Box \exists x Qx$  are not "clean", whereas $\exists x \Box Px \lor \forall y \Diamond Qy$ and 
$Px \land  \Box \exists y Qy$ are their "clean" equivalents respectively. Note that every 
$\FOML$-formula can be rewritten into an equivalent "clean formula".  A finite set of formulas $\Gamma$ is "clean" if $\andSet$ is "clean".
"Clean formulas" help in handling the witnesses for existential formulas in the tableau in a syntactic way.  

Consider a finite set of formulas $\Gamma$ that is "clean". Suppose we want to expand $\Gamma$ to $\Gamma \cup \{\phi_1,\ldots, \phi_k\}$,  then even if each of $\phi_i$ is "clean", it is possible that a bound variable of $\phi_i$  also occurs in some $\psi \in \Gamma$ or another $\phi_j$. To avoid this, first, we rewrite the bound variables in each  $\phi_i$ one by one by using the fresh variables that do not occur in $\Gamma$ and other previously rewritten $\phi_j$.
Such a rewriting can be fixed by always using the first fresh variable in a fixed  enumeration of all the variables. When $\Gamma$ and $\{\phi_1,\ldots \phi_k\}$ are clear from the context, we denote $\phi_i^*$ to be one such rewriting of $\phi_i$ into a "clean formula". It is not hard to see that the resulting finite set $\Gamma \cup \{\phi_1^*,\ldots \phi_k^*\}$ is "clean". 

\bigskip
\AP
 In a "tableau" node of the form $(w:\Gamma,S,F)$, we allow $F$ to be either a "skolem-forest" or $F = \intro*\notInit$ or $F = \intro*\emptyTree$. If $F = \notInit$, it indicates that $F$ is not yet initialized and if $F = \emptyTree$ then $F$ is an empty tree. We assign $F = \emptyTree$ when $\Gamma$ is "nested $\forall$ free" and hence no skolem witnesses are needed.

  \begin{figure*}
 \begin{center}
 \begin{tabular}{|c c|}
   
   \hline
\multicolumn{2}{|c|}{} \\
     $\dfrac{w:\{\phi_1\lor\phi_2\}\cup \Gamma,S,F}{w: \{\phi_1\}\cup\Gamma,S,F~~~\mid~~~ w:\{\phi_2\}\cup \Gamma,S,F}\ \intro*\orRule$ &  $\dfrac{w:\{\phi_1\land\phi_2\}\Gamma,S,F}{w:\{\phi_1,\phi_2\}\cup \Gamma,S,F} \ \intro*\andRule$ \\&\\
    where $F \ne \notInit$& where $F\ne \notInit$\\
     \hline
     \hline

      \multicolumn{2}{|c|}{}\\
      \multicolumn{2}{|c|}{$\dfrac{w:\{\exists y~ \phi\}\cup \Gamma,~S,~F}
       {w: \{\phi\}\cup \Gamma,~S,~F} \intro*\existsRule$} \\& \\
      \multicolumn{2}{|c|}{where $F\ne \notInit$}\\
      \hline
      \hline
      
       \multicolumn{2}{|c|}{}\\
      \multicolumn{2}{|c|}{$\dfrac{w:\{\forall y~\phi\}\cup \Gamma,S,F}
       {w:~\{\phi^*[y/z] \mid z \forestVariable F\} \cup \Gamma,~S,~F}\intro*\forallRule$} \\& \\
        %\multicolumn{2}{|c|}{ where $\Gamma$ is a "module set", "nested $\forall$ free" and "outermost $\exists$ free"}\\
        \multicolumn{2}{|c|}{where $F\ne \notInit$ and }\\
        \multicolumn{2}{|c|}{every $\phi^*[z/y]$ is a clean rewriting of $\phi[y/z]$ "wrt" $\Gamma \cup \{\phi[y/z]\mid z\forestVariable F \}$}\\
      \multicolumn{2}{|c|}{} \\ 
      \hline
      \hline
      
       \multicolumn{2}{|c|}{Given $n\geq 1$ and $m,s\geq 0$} \\
      \multicolumn{2}{|c|}{}\\
       \multicolumn{2}{|c|}{$\cfrac{\genfrac{}{}{0pt}{0}{w:\{\Diamond \phi_1,\ldots,\Diamond \phi_n\}\cup }{\{\Box \alpha_1,\ldots \Box \alpha_m\}\cup} \ \{l_1,\cdots,l_s\},~S,~F}{\left\langle wv_i: \{ \phi_i\}~\cup \{\alpha_j \mid j\le m\},~S_i,~\notInit\right\rangle \ \text{for all} \ i\le n } \quad \intro*\diamondRule$} \\  & \\
       \multicolumn{2}{|c|}{where each $l_k$ is a "literal" and $F \ne \notInit$}\\
       \multicolumn{2}{|c|}{For every $i$, $S_i = S \cup \outerExVar[\phi_i] ~ \cup~~ \bigcup\limits_{i\le m}\outerExVar[\alpha_i]$} \\
      \hline
      \hline

      \multicolumn{2}{|c|}{}\\
      \multicolumn{2}{|c|}{$\dfrac{w:\Gamma ,~S,~\notInit}
       {w: \Gamma',~S',~F} \intro*\nestedForallRule$} \\& \\
      \multicolumn{2}{|c|}{where $\Gamma$ contains a unique "nested $\forall$ formula" $\forall x~\psi$,}\\
      \multicolumn{2}{|c|}{$F$\text{ is a "skolem-forest" } "wrt" $(\Gamma,S)$ and $\Gamma' = \expandedGamma$ is the "expansion of skolem-forest" $F$}\\
      \multicolumn{2}{|c|}{$S'$ is the set of all variables/nodes in the "skolem-forest" $F$}\\
      \hline
      \hline

      \multicolumn{2}{|c|}{}\\
      \multicolumn{2}{|c|}{$\dfrac{w:\Gamma,~S,~\notInit}
       {w: \Gamma,~S,\emptyTree} \intro*\trivialSkolemRule$} \\& \\
      \multicolumn{2}{|c|}{where $\Gamma$ is "nested $\forall$ free"}\\
      \hline
      \hline
      
      \multicolumn{2}{|c|}{Given $m,s\geq 0$:} \\
      \multicolumn{2}{|c|}{} \\
      \multicolumn{2}{|c|}{$\dfrac{w:\{\Box\beta_1,\ldots \Box \beta_m\} \cup \{l_1,\ldots,l_s\},~S,~F}{w:\{l_1,\cdots,l_s\},~S,~F} \quad \intro*\EndRule$} \\&\\
      \multicolumn{2}{|c|}{where $F\ne \notInit$ and each $l_k$ is a "literal"} \\
      \hline
       
\end{tabular}
\caption{Tableau rules for $\ourEBBE$}
%, here every $l_i$ is a "literal"}
\label{fig-tableau for EBBE}
\end{center}
\end{figure*}

\bigskip
The tableau rules for the $\ourEBBE$-fragment are described in Fig. \ref{fig-tableau for EBBE}. 
 The $\andRule$ and $\orRule$ are standard, where we make a non-deterministic choice of one of the branches for $\orRule$. The  $\EndRule$ says that  if we are left with only modules and there are no $ \Diamond$ formulas, then the branch does not need to be explored further. The $\diamondRule$ creates one successor world for every $\Diamond$ formula at the current node and includes all the $\Box$ formulas that need to be satisfied along with the $\Diamond$ formula. In this case, $S$ is inherited in the successor worlds to preserve the increasing domain property and $F$ is set to $\notInit$ to {\em `forget'} the skolem-witness information of the parent. The $\existsRule$ rule picks $y$ itself as the witness to satisfy $\exists y \phi$. Note that in this case,by construction, $y\in S$ and hence we are not introducing any new variable. The $\forallRule$ rule expands the set of formulas to include a clean version of $\phi[z/y]$ for every variable $z \in S$ where $S$ is intended to be the "local domain". 
 
 The $\nestedForallRule$ is novel. In this rule, the tableau non-deterministically picks a "skolem-forest" "wrt" $\Gamma$ and $\forall x~\psi$. $\Gamma$ is then replaced by the $\expandedGamma$ which removes this unique "nested $\forall$ formula" $\forall x ~\psi$ and replaces it with a finite set of "nested $\forall$ free" formulas. Note that for this rule to be well defined, we need to ensure that $\forall x~\psi$ is the unique "nested $\forall$ formula" in $\Gamma$ (otherwise, the "skolem-forest" is not well defined). But this follows from construction since the root cannot contain a "nested $\forall$ formula" and every "nested $\forall$ formula" has be be of the form $\Diamond \forall x~\psi$ where we would have applied a $\diamondRule$ in the parent and we will have no other "nested $\forall$ formula" that is inherited from the parent.
 The $\trivialSkolemRule$ initializes the "skolem-forest" to $\emptyTree$ to asserts that no skolem witnesses are needed.

 Note that there is an implicit ordering on how the rules are applied: $\diamondRule$ can be applied at a node $(w: \Gamma,S,F)$ only if $\Gamma$ contains only "modules" and hence may be applied only after all other rules have been applied as  many times as necessary at $w$. For a new world label $w$, initially $F = \notInit$ always (from $\diamondRule$) and hence we can only apply either $\trivialSkolemRule$ or $\nestedForallRule$ depending on whether $\Gamma$ is "nested $\forall$ free" or not "respectively". Note both $\trivialSkolemRule$ and $\nestedForallRule$ always result in a set $\Gamma$ that is "nested $\forall$ free".
 Moreover, if $F\ne \notInit$ then we have already applied either $\trivialSkolemRule$ or $\nestedForallRule$ for some node in the tableau with world label $w$ in the past. Hence, whenever $F \ne \notInit$, then $\Gamma$ is always "nested $\forall$ free". After this point, we will only have to deal with "nested $\forall$ free" formulas. Moreover, the set $S$ resulting after this step will remain the same until the "world" name remains $w$.

 Note that the $\diamondRule$ is the only rule that can change (the name of) the possible world, thus creating a new successor "world". It simply extends the name $w$ by new symbols $v_i$ for each successor. Therefore there can be many nodes in the tableau sharing the same world name, but such nodes form a path. Given $w$ we use $\intro*\lastW$ to denote the ""last node"" down the "tableau", whose first component is $w$. 

\begin{propositionrep}
\label{prop-some rule can be applied always}
For every "tableau" $T$ and every node $\tau= (w: \Gamma,S,F)$ in $T$, if $\tau$ is a leaf in $T$ then $F\ne \notInit$ and moreover, either $\Gamma$ contains only "literals" or there is some rule that can be applied at $\tau$.
\end{propositionrep}
\begin{proof}
Assume that $\tau = (w:\Gamma,S,F)$ is a leaf not but contains a formula that is not a "literal".

First, if $F = \notInit$ then we can always apply either $\nestedForallRule$ or $\trivialSkolemRule$ depending on whether $\Gamma$ is "nested $\forall$ free" or not. So assume that $F \ne \notInit$.

Now suppose $\Gamma$ contains at least one formula of the form $\phi\land \psi$ or $\phi \lor \psi$ then we can apply the $\andRule$ or $\orRule$ "respectively". So assume that the operators $\land, \lor$ do not occur as the outermost connective in any formula of $\Gamma$. Further, if every $\phi\in \Gamma$ is a 
"module" then we can apply $\diamondRule$ if  at least one formula of the form $\Diamond \phi\in \Gamma$, otherwise we can apply $\EndRule$.

So the remaining case is that there is at least one formula in $\Gamma$ which has a quantifier. Now if we have some $\exists x~ \phi\in \Gamma$ then we can apply $(\exists)$ rule. Otherwise, every quantified formula in $\Gamma$ is of the form $\forall y \phi$ and we can apply the $\forallRule$.
\end{proof}

Given a formula $\theta \in \ourEBBE$ the "tableau" starts with $(r,\{\theta\}, S,\notInit)$ where $S = \FV[\theta] \cup \{z\}$ where $z$ is a fresh variable not occurring in $\theta$.

\begin{theorem}
\label{thm-tableau for EBBE}
For any clean $\ourEBBE$ formula $\theta$, let $S_0 = \FV[\theta]\cup \{z\}$ where $z$ does not occur in $\theta$. Then:

There is an "open tableau" $T$ with root $(r:\{\theta\},S_0,\notInit)$ iff $\theta$ is satisfiable in an "increasing domain model".
\end{theorem}

The goal of the rest of the paper is to prove \autoref{thm-tableau for EBBE}.

\section{Soundness}
\label{sec-soundness}
Before we prove  $(\Leftarrow)$ of \autoref{thm-tableau for EBBE}, we have the following useful lemma.

\begin{lemmarep}
\label{lemma-skolem-forest-has-bounded-size}
Let $\Gamma$ be a finite set of "$\ourEBBE$ subformulas" that contains a unique  "nested $\forall$ formula" $\forall x~\psi$ and let $S = \FV[\Gamma] \cup \outerExVar[\Gamma]$. If $\Gamma$ is "satisfiable", then there is a "skolem-forest" $\skolemForest[\Gamma,S]$ of size at most $n\cdot m^{2^{O(m)}}$ where $|S|=n$ and $|\phi| = m$. Moreover, there is some  "expansion of skolem-forest" $\expandedGamma$ such that $\expandedGamma$ is "satisfiable".
\end{lemmarep}
\begin{proof}
    Let $\M,w,\sigma \models \andSet$ and let $S = \FV[\Gamma] \cup \outerExVar[\Gamma]$ where $\outerExVar[\Gamma] = \{ y_1\ldots y_k\mid k\ge 0\}$. For every $y_i\in \outerExVar[\Gamma]$ we have a formula of the form $\exists y_i~\alpha_i\in \Gamma$ and there exists some $d_i\in \live(w)$ such that $\M,w,\sigma_{[y_i\rightarrow d_i]} \models \alpha_i$.

To construct the "skolem-forest" $F$, we define the nodes of all trees in $(V,E)$ by induction of the height of the nodes. First, for every $z\in S$ we create a new tree $T_z$ rooted at $z$. Now let $\hat{z} = \sigma(z)\in \live(w)$ if $z\in \FV[\Gamma]$ and if $z = y_i\in \outerExVar[\Gamma]$ then let $\hat{z} = d_i\in \live(w)$. Now since $\M,w,\sigma \models \forall x~\psi$, for every $d\in \live(w)$, we have $\M,w,\sigma_{[x\to \hat{z}]} \models \psi$. Define $\forestLabel(z) = \{ \zeta \mid \zeta\in \compPlus[\psi]$ and $\M,w,\sigma_{[x\to \hat{z}]} \models \zeta\}$. Thus, there is some "atom" $\atom[\psi]$ of $\psi$ such that $\forestLabel(z) = \atom[\psi]$ "wrt" $x$. Note that by construction, for every $\zeta\in \forestLabel(z)$ we have $\M,w,\sigma_{[x\to \hat{z}]} \models \zeta$.

\bigskip
    Assume that inductively we have constructed some tree $T_r$ rooted at $r$ upto depth $h$ and we are at variable $z$. Suppose $\forestLabel(z)$ does not contain any formula of the form $\exists y~\alpha$ or there is some two distinct variables $z_1,z_2\ne z$ along the path from $r$ to $z$ in $T_r$ such that  $\forestLabel(z_1) = \forestLabel(z_2) = \forestLabel(z)$, then we are done.  Otherwise, for every $\exists y \alpha \in \forestLabel(z)$, create a new child $y_z$ of $z$ in $T_r$ where $y_z$ is fresh. It remains to define $\forestLabel(y_z)$ which is done as follows:
    
    Since $\exists y \zeta \in \forestLabel(z)$, by the inductive construction, we have $\M,w,\sigma_{[x\to \hat{z}]} \models \exists y \zeta$. Hence there is some $c\in \live(w)$ such that $\M,w,\sigma_{[x\to \hat{z}, y\to c]}\models \zeta$. Furthermore, we also have $\M,w,\sigma_{[x\to c]}\models \psi$ (since $\M,w,\sigma \models \forall x \psi$).  
    Let $\forestLabel(y_z) = \{ \zeta \mid \zeta\in \compPlus[\psi]$ and $\M,w,\sigma_{[x\to c]} \models \zeta\}$ and let $\hat{y_z} = c$. 

\medskip
    Note that since $\forestLabel(z)\subseteq \compPlus[\psi]$, that there are only $2^{O(m)}$ many possible $\forestLabel(z)$. Hence every path in the forest can be of height at most $2^{O(m)}$. Moreover, every vertex can have at most $O(m)$ many children. Hence, the total number of nodes in $\skolemForest$ is $n\cdot m^{2^{O(m)}}$. So we have the "skolem-forest" $F$ of the required size.

Finally, define "expansion of skolem-forest" $\expandedGamma = (\Gamma \cup \Gamma') \setminus \{\forall x~\psi\}$ where
$\Gamma' = \bigcup\limits_{z\in V} \{ \beta[z/x]\mid \beta$ is not of the form $\exists y \alpha$ and $\beta\in \forestLabel(z)\} \cup \bigcup\limits_{z\in V\\\text{where }z\text{ is not a leaf}}\{ \alpha[z/x,y/y_z] \mid \beta$ is of the form $\exists y \alpha\}$.
By the definition of $\Gamma'$ it follows that the "expansion of skolem-forest" $\expandedGamma$ satisfies all the required properties. 

To prove that $\expandedGamma$ is "satisfiable", if $\phi \in \Gamma\setminus \Gamma'$ then $\M,w,\sigma \models \phi$ and if $\phi\in \Gamma'$ and it is not of the form $\exists y~\beta$ then $\phi$ is of the form $\phi[z/x]$ for some $\phi\in \compPlus[\psi]$. Now $\M,w,\sigma_{[x\to \hat{z}]}\models \phi(x)$ and hence $\M,w,\sigma_{[z\to \hat{z}]}\models \phi[z/x]$. If $\phi$ is of the form $\exists y~\beta[x/z,y_z/y]$ and by construction we have $\M,w,\sigma_{[z\to \hat{z}, y\to \hat{y_z}]}\models \beta[z/x,y/y_z]$. Hence $\expandedGamma$ is "satisfiable".
\end{proof}
\begin{proofsketch}
 Since $\Gamma$ is "satisfiable", there is some $\M$ and some $w\in \W$ and an appropriate $\sigma$ such that $\M,w,\sigma \models \Gamma$. For every $z\in S$, we make $z$ as a root and assign them the "atom" that corresponds to what $\sigma(z)$  prove in that in the same model $\Gamma'$ will also be valid. We keep on extending the trees by looking at some domain element appropriately to define the $\forestLabel(z)$. The tree size is `small' because every node has boundedly many children to consider and the depth of the tree is bounded because $\forestLabel(z)$ will repeat eventually. For full proof, please refer \Cref{sec-soundness-appendix}.
\end{proofsketch}

\begin{toappendix}
\label{sec-soundness-appendix}
\begin{proof}[Proof of $(\Leftarrow)$ of \autoref{thm-tableau for EBBE}]
From \Cref{prop-some rule can be applied always} it follows that we can always apply some rule until every leaf node $(w: \Gamma,S,F)$ is such that $\Gamma$ contains only "literals". Thus every (partial) tableau can be extended to a "saturated tableau". To prove that such a "tableau" is "open", it  suffices to show that all rules preserve "satisfiability". 
"ie" we prove that for every rule, if the antecedent $(w: \Gamma, S,F)$ is such that $\Gamma$ is "satisfiable", then every new node in the consequent $\langle v : \Gamma', S',F' \rangle$ the set $\Gamma'$ is also "satisfiable" (and hence cannot violate the "open" condition).

\begin{itemize}
\item The case of $\andRule$ and $\EndRule$ are trivial. 

\item For $\orRule$, if $\Gamma \cup\{\phi_1\lor\phi_2\}$  is satisfiable, then $\Gamma\cup \{\phi_1\}$ is satisfiable or $\Gamma\cup \{\phi_2\}$ is satisfiable. The non-deterministic guess of the $\orRule$ ensures that $\Gamma'$ is "satisfiable". 

\item For $\existsRule$,  $\Gamma = \Gamma' \cup \{ \exists y~\phi\}$ where the $\existsRule$ is applied for $\exists y~\phi$. Suppose $\M=(\W,\D,\live,\R,\val)$ is an "increasing domain model", $w\in \W$ and $\pi$ is an assignment such that  $\M,w,\pi\models\exists y\phi\land \andSet[\Gamma']$. By semantics, it follows that there is a element $d\in \live(w)$ such that $\M,w,\pi_{[y\mapsto d]}\models\phi$. By "cleanliness", $y$ is not free in  $\Gamma$ and hence $\M,w,\pi_{[y\mapsto d]}\models\andSet$. Thus, $\{\phi\}\cup\Gamma'$  is also "satisfiable". 

\item For $\forallRule$ rule, let $\Gamma = \Gamma' \cup \{ \forall y~\phi\}$ suppose $\M,w,\pi\models\forall y \phi\land \andSet[\Gamma']$.  By semantics, for every $d\in \live(w)$ we have  $\M,w,\pi_{[y\mapsto d]}\models\phi $. Let $S = \{z_0, z_1,\ldots z_k\}$ and for all $i\le k$ let $\pi(z_i) = d_i$. 

Now since every $\phi^*[z/y]$ is a clean rewriting of $\phi[z/y]$, we have $\M,w,\pi_{[z_0\mapsto d_0,\ldots z_k\mapsto d_k] }\models \phi^*[z_0/y]\land \ldots \phi^*[z_k/y]$. Also, we also have $\M,w,\pi_{[z_0\mapsto d_0,\ldots z_k\mapsto d_k]} \models \andSet[\Gamma']$.
Hence the set of formulas $\{\phi^*[z/y]\mid z\in S\}\cup \Gamma'$ is "satisfiable".  

\item For $\diamondRule$, suppose $\M,w,\pi\models \Diamond \phi_1 \land \ldots\land \Diamond \phi_n \land \Box \alpha_1 \land \ldots \land \Box \alpha_m\land l_1\land \ldots \land l_s$ then pick an arbitrary successor world $wv_i$. For every $\phi_i$ and we have $w\rightarrow u\in \R$ and $\M,u,\sigma \models \phi_i \land \alpha_1\land \ldots \alpha_m$. Hence $\Gamma'$ is "satisfiable".

\item For $\nestedForallRule$, suppose $M,w,\pi \models \hat\Gamma \land \forall x~\psi$ where $\forall x~\psi$ is a "nested $\forall$ formula", then by \Cref{lemma-skolem-forest-has-bounded-size}, we can always have the required "skolem-forest" $F$ of bounded size with the corresponding $\Gamma'$ to be satisfiable.

\item For $\trivialSkolemRule$, the claim is obvious since $\Gamma$ remains the same.

\end{itemize}
\end{proof}
\end{toappendix}

To prove $(\Leftarrow)$ of \autoref{thm-tableau for EBBE}, from \Cref{prop-some rule can be applied always} it follows that we can always apply some rule until every leaf node $(w: \Gamma,S,F)$ is such that $\Gamma$ contains only "literals". Thus every (partial) tableau can be extended to a "saturated tableau". To prove that such a "tableau" is "open", it  suffices to show that all rules preserve "satisfiability". 
"ie" we prove that for every rule, if the antecedant $(w: \Gamma, S,F)$ is such that $\Gamma$ is "satisfiable", then every new node in the consequent $\langle v : \Gamma', S',F' \rangle$ the set $\Gamma'$ is also "satisfiable" (and hence cannot violate the "open" condition).

For all rules except the $\nestedForallRule$, the claim follows along the standard lines and for $\nestedForallRule$, we use \Cref{lemma-skolem-forest-has-bounded-size} to prove the claim. 
Formal proof is provided in \Cref{sec-soundness-appendix}.

\section{Completeness}
\label{sec-completeness}
 Let $\intro*\identity:\Var\to \Var$ be the identity mapping.  Let $T$ be an "open tableau". For every node $\tau = (w:\Gamma,S,F)$ we denote $\tau_w,\tau_\Gamma, \tau_S,\tau_F$ to denote the respective components and let $\lastW = (w:\Gamma_w,S_w,F_w)$ be the "last node" with world label $w$. 
 
Suppose $T$ is an "open tableau" that contains some node of the form $\tau = (w:\Gamma,S,F)$ obtained using the $\nestedForallRule$ where $F$ is a "skolem-forest" and let $z$ be a leaf node in some tree of $F$ such that $\exists y~\alpha \in \forestLabel(z)$. From the construction it means that we want a witness for $z$ but the "skolem-forest" is not providing one. The next definition shows how to extend the "skolem-forest" so that $z$ will have the witness.

\begin{definition}[Tableau extension]
\label{def-tableau-extension} 
Let $T$ be an "open tableau" and $\tau = (w:\Gamma,S,F)$ be a node in $T$ which is obtained by applying $\nestedForallRule$ and let $z$ be a leaf node in some tree of $F$ such that there is some $\exists y \alpha\in \forestLabel(z)$. Let $T_\tau$ be the subtree of $T$ rooted at $\tau$. We say $T'$ ""extends"" $T$ at $\tau$ for $z$ if the following conditions hold:
\begin{enumerate}
    \item All vertices $\tau_1\in T\setminus T_\tau$ continue to be a node in $T'$.
    \item The node $\tau$ in $T$ is replaced with $\tau' = (w:\Gamma',S',F')$ where $F'$ extends $F$ and $z$ is not a leaf node such that $\tau'$ is a valid result of applying $\nestedForallRule$ to the predecessor of $\tau$ in $T$.
    \item $\{u\mid (u:\Gamma_1,S_1,F_1) \in T_\tau\} = \{v\mid (v: \Gamma_1',S_1',F_1') \in T'_{\tau'}\}$.
\end{enumerate}
\end{definition}

Note that, in the second condition, any arbitrary extension of the "skolem-forest"  $F$ to $F'$ is allowed as long as $z$ is not a leaf node and $\tau'$ satisfies all the required conditions for the $\nestedForallRule$. The first condition states that the extended tree does not change the nodes of $T$ that are not in the subtree $T_\tau$. The second condition states that $z$ is not a leaf node and hence, for every $\exists y~\alpha\in \forestLabel'(z)$ we will be able to find a corresponding witness in the "skolem-forest" $F'$.  But then this means that we will have more formulas to satisfy. The third condition ensures that the "worlds" that are needed to satisfy the formulas remains the same.
The following lemma states that it is always possible to extend an "open tableau".

\begin{lemma}
\label{lemma-skolem-tree-extension}
Let $T$ be an "open tableau" and let $\tau = (w:\Gamma,S,F)$ be a node in $T$ obtained using the $\nestedForallRule$ and let $z$ be a leaf node in $F$. Then there exists an "open tableau" $T'$ that extends $T$ at $\tau$ for $z$.
\end{lemma}
\begin{proofsketch}
Pick a node $\tau = (w:\Gamma,S,F)$ as specified and let $z$ be a leaf node in one of the trees $T$ of $F$ such that $\exists y~\alpha \in \forestLabel(z)$. Let $r$ be the root of $T$. Since $F$ is a "skolem-forest", there are two distinct $z_1,z_2\ne z$ along the path from $r$ to $z$ in $T$ such that $\forestLabel(z_1) = \forestLabel(z_2) = \forestLabel(z)$. Let $T_{z_1}$ be the subtree of $T$ rooted at $z_1$. Create a new copy of $T_{z_1}$ (call it $T'_{z_1}$) and replace the root of $T'_{z_1}$ with $z$ and all other descendants with a fresh variable.
Attach $T'_{z_1}$ at $z$ in $T$ to get $T'$ so that the subtree rooted at $z$ in $T'$ is $T'_{z_1}$. For every new variable $y\in T'$ define $\forestLabel'(y') = \forestLabel(y)$ where $y'$ is fresh and is a copy of some $y$ in $T$. For all the old variables $y \in T'$ keep $\forestLabel'(y) = \forestLabel(y)$.

In the new tree $T'$, $z$ is not a leaf anymore and every new variable $y'$ in $T'_{z_1}$ is a copy of $y$ in $T_{z_1}$.  $F'$ obtained from $F$ by replacing $T$ with $T'$. It can be verified that $F'$  satisfied all conditions to be a "skolem-forest" since the old leafs continue to satisfy the conditions and we can argue that if some new leaf violates the condition then some old leaf in $T$ was already violating the condition. 

\sloppy The new "expansion of skolem-forest" $\expandedGamma'$ can be constructed from $\expandedGamma$ by copying $\Gamma$ for the new variables with corresponding substitution. "ie" $\expandedGamma ' = \expandedGamma \cup \{ \alpha[y'/y]\mid y'$ is a new variable which is a copy of $y$ and $\alpha(y)\in \expandedGamma$ and $\alpha(y)$ is not of the form $\exists y_2~\beta\} \cup \{ \exists y_2~\beta[y_1'/y_1, y_2'/y_2]\mid \exists y_2~\beta[y_1,y_2]\in \expandedGamma$ and $y_1',y_2'$ are the copies of $y_1,y_2$ "respectively"$\}$. Again it is can be verified that the new $\expandedGamma'$ is indeed an "expansion of skolem-forest" $F'$ (otherwise we can argue that the original $\expandedGamma$ itself would is not an "expansion of skolem-forest"). 

To complete the tableau, in every descendant $\tau' = (w': \Gamma',S',F')$ of $\tau$ in $T$ (where $w'\ne w$), we define new node $\tau'_1 = (w' : \Gamma'_1,S'_1,F'_1)$ where $S_1' = S' \cup \{ y\mid y$ is added fresh$\}$ and if $F'_1\ne \notInit,\emptyTree$ then $F_1'$ extends $F'$ by adding a new tree rooted at $T_{y'}$ for every fresh variable added, where $T_{y'}$ is an isomorphic copy of $T_y$ where $y'$ is a copy of $y$. 

Finally we will need to add intermediate nodes between successive nodes in the descendants of $\tau$ to ensure that the corresponding rules have been applied to the newly added formulas. This is proved formally by setting up a tedious but routine induction. 
\end{proofsketch}

  For any arbitrary "open tableau" $T$ we associate it with a model $\M_T = (\W,\D,\R,\live,\val)$ where:
\begin{itemize}
    \item $\W=\{w\mid (w: \Gamma,\sigma)$ is a node in $T\}$
    \item $\D=\Var$
    \item $\R = \{ (w,v) \mid v$ is of the form $wv'$ for some $v'\}$
    \item For every $w\in \W$, if $F\ne \emptyTree$ then $\live(w) = \{z\mid z\forestVariable F_w\}$, otherwise $\live(w) = S_w$ where "last node" of $w$ is $(w:\Gamma_w,S_w,F)$.
    \item For every $w\in \W$ and $P\in \Ps$,  define 
    $\val_i(w,P) = \{ (x_1,\ldots x_k) \mid P(x_1\ldots x_k) \in \Gamma_w\}$ where "last node" of $w$ is $(w:\Gamma_w,S_w,F)$.
\end{itemize}

Typically at this point we prove the truth lemma to complete the proof that states for any node $(w:\Gamma,S,F)$ if $\phi\in \Gamma$ then $\M_T,w,\identity \models \phi$ but in this case the claim will not hold in general. However we can pin point where the problem is. 
 The next lemma says that if this property is violated, then the problem occurs only because of some "nested $\forall$ formula" while trying to satisfy an existentially quantified subformula.

\begin{lemmarep}
\label{lemma-leaf-violation}
Let $T$ be an "open tableau" and $\M_T$ be the associated model. For every node $\tau = (w:\Gamma,S,F)$  in $T$ and every $\phi\in \Gamma$ if $\M_T,w,\identity \not\models \phi$ then there exists some $\forall x~\psi \in \SFplus[\phi]$ such that $\forall x~\psi$ is a "nested $\forall$ formula" and there exists some descendant node $\tau' = (w': \Gamma',S',F')$ of $\tau$ in $T$ such that
 $\exists y \beta \in \compPlus[\psi]$ and some variable $z \in \live(w')$ such that $z$ is a leaf node in some tree of $F'$ with $\exists y \beta\in \forestLabel(z)$ and $\M,w',\identity_{[x\mapsto z]} \not\models \exists y \beta$.
\end{lemmarep}
\begin{proof}
\label{completeness-appendix}
Let $(w:\Gamma,S,F)$ and $\phi\in \Gamma$ such that $\M_T,w,\identity \not\models \phi$.  The proof is by reverse induction on the height of the node $(w:\Gamma,S,F)$ in $T$.  For the base case, $(w:\Gamma,S,F)$ is a leaf node, and hence $\Gamma$ contains only "literals". Thus, if $P(x_1,\ldots x_n)\in \Gamma$ then by the definition of $\val_i(w,P)$ we have $(x_1\ldots, x_n) \in \live(w,P)$ and hence$\M_T,w,\identity \models P(x_1\ldots x_n)$. So in this case the claim is vacuously true.

For the inductive step, assume that $(w:\Gamma,S,F)$ is a non-leaf node in $T$. Hence some rule is  applied at this node and we have one or more descendants depending on the rule. Also, the claim holds for all the descendant nodes. Now we consider various cases depending on which rule was applied.

\begin{itemize}
\item If $\andRule$ was applied then $\Gamma$ is of the form $\Gamma'\cup \{\phi_1\land \phi_2\}$ and the node $\tau = (w:\Gamma,S,F)$ has one descendant $\tau_1 = (w: \Gamma'\cup\{\phi_1,\phi_2\},S,F)$. Now if $\phi\in \Gamma'$ then inductively the claim holds at the $\tau_1 = (w: \Gamma'\cup\{\phi_1,\phi_2\},S,F)$, which means we have a descendant $\tau'$ of $\tau_1$ as required. But then $\tau'$ is also is also a descendant of $\tau$ and we are done.

Otherwise $\phi = \phi_1\land \phi_2$. By assumption we have $\M_T,w,\identity\not\models \phi_1\land \phi_2$ and hence there is some $\phi_i$ where $i\in \{1,2\}$ such that $\M_T,w,\identity \not\models \phi_i$. But then $\phi_i$ is also a formula in the descendant $\tau_1 = (w: \Gamma'\cup\{\phi_1,\phi_2\},S,F)$ and hence the claim holds at $\tau_1$. This means that we have a descendant of $\tau'$ of $\tau_1$ as required. But then $\tau'$ is also is also a descendant of $\tau$ and we are done.

\item If $\orRule$ is similar to the $\andRule$.
%was applied then $\Gamma$ is of the form $\Gamma'\cup \{\phi_1\lor \phi_2\}$ and the node $(w:\Gamma,S,F)$ has one descendant $\tau_1 = (w: \Gamma'\cup\{\phi_i\},\sigma)$ for some $i\in \{1,2\}$.  Again if $\phi\in \Gamma'$ then the claim follows by induction. Otherwise, for both $i\in \{1,2\}$ we have $\M_T,w,\identity \not\models \phi_i$. But then at least one $\phi_i$ is a formula in the descendant $\tau_1$ and hence the claim holds at $\tau_1$. This means that we have a descendant of $\tau'$ of $\tau_1$ as required. But then $\tau'$ is also is also a descendant of $\tau$ and we are done.

\item If $\existsRule$ was applied then $\Gamma$ is of the form $\Gamma'\cup \{\exists y~ \phi_1\}$ and the node $(w:\Gamma,S,F)$ has one descendant $\tau_1 = (w:\Gamma'\cup \{\phi_1\},S,F)$. Again if $\phi\in \Gamma'$ then the claim follows by induction. Otherwise,  $\phi$ is $\exists y~\phi_1$. Note that by construction, $y\in \live(w)$ which implies that the claim holds at $\tau_1$. Hence, we have a descendant of $\tau'$ of $\tau_1$ as required. But then $\tau'$ is also is also a descendant of $\tau$ and we are done.

\item If $\forallRule$ is applied at $(w: \Gamma,S,F)$ then $\Gamma$ is of the form $\Gamma'\cup \{\forall y~ \phi_1\}$ and we have a descendant $\tau_1 = (w: \Gamma' \cup \{\phi^*[z/y]\mid z\forestVariable F\}, S,F)$. Again if $\phi\in \Gamma'$ then the claim follows by induction. Otherwise, $\phi$ is $\forall y~\phi_1$. So for some $z\in \live(w)$ we have $\M,w,\identity \not\models \phi'[z/x]$. But then, the claim holds at $\tau_1$. This means that we have a descendant of $\tau'$ of $\tau_1$ as required. But then $\tau'$ is also is also a descendant of $\tau$ and we are done.

\item If $\diamondRule$ is applied at $(w:\Gamma,S,F)$ then $\Gamma$ is of the form $\{l_1\ldots l_s\} \cup \{\Diamond\phi_1\ldots \Diamond\phi_n\}\cup \{\Box\alpha_1\ldots \Box\alpha_m\}$ for some $s,m\ge 0$ and $n\ge 1$.

As argued in the base case, $\phi$ cannot be a "literal". If $\phi$ is of the form $\Diamond \phi_i$ then there is a child $\tau_1 = (wv_i: \Gamma',S',\notInit)$ of $\tau$ such that $\phi_i \in \Gamma'$ but by assumption, $\M_T,wv_i,\identity \not\models \phi_i$. Hence the claim holds at $\tau_1$. This means that we have a descendant of $\tau'$ of $\tau_1$ as required. But then $\tau'$ is also is also a descendant of $\tau$ and we are done.

On the other hand,  if $\phi$ is of the form $\Box \alpha_j$,  then there at least one child  $\tau_1= (wv_i: \Gamma',S',\notInit)$ of $\tau$ such that $\M_T,wv_i,\identity \not\models \alpha_j$ but by $\diamondRule$, $\alpha_j \in \Gamma'$. Hence the claim holds at $\tau_1$. This means that we have a descendant of $\tau'$ of $\tau_1$ as required. But then $\tau'$ is also is also a descendant of $\tau$ and we are done.

\item If $\nestedForallRule$ is applied at $(w:\Gamma,S,F)$ then $F = \notInit$ and $\Gamma$ has a unique $\forall x~\psi$ which is a "nested $\forall$ formula". Let the child of $\tau$ be $\tau_1 = (w,\Gamma',S',F')$. Again if $\phi\in \Gamma'$ then the claim follows by induction. Otherwise $\phi$ is the "nested $\forall$ formula" $\forall x~\psi$. Now since $\M,w,\identity \not\models \forall x~\psi$ there is some $z\in \live(w)$ such that $\M,w,\identity_{[x\to z]} \not\models \psi$. This implies that there is some $\psi' \in \forestLabel(z)$ such that $\M,w,\identity_{[x\to z]} \not\models \psi'$. Now, again as argued in the base case, $\psi'$ cannot be a "literal" and as argued in the $\diamondRule$, $\psi'$ cannot be a "module". Now, if $\psi'$ is of the form $\forall y \beta$ then by construction, we have $\beta[z/x] \in \Gamma'$ and hence the claim holds at $\tau_1$. This means that we have a descendant of $\tau'$ of $\tau_1$ as required. But then $\tau'$ is also is also a descendant of $\tau$ and we are done.

Finally, if $\psi'$ is of the form $\exists y \beta$ and $z$ is a non-leaf node then there is some child $y_z$ of $z$ such that $\exists y_z\beta[x/z,y/y_z] \in \Gamma'$ and hence the claim holds at $\tau_1$. This means that we have a descendant of $\tau'$ of $\tau_1$ as required. But then $\tau'$ is also is also a descendant of $\tau$ and we are done. \\
Otherwise, $z$ is a leaf node and $\psi'$ is of the form $\exists y \beta$. But then this implies that the claim is true at $\tau$ itself and we are done.

\item If $\EndRule$ rule is applied at $(w:\Gamma,\sigma)$ then $\Gamma$ is of the form  $\{l_1\ldots l_s\} \cup \{\Box\beta_1\ldots \Box\beta_m\}$ for some $s,m\ge 0$. There is one descendant $(w:\{l_1\ldots l_s\},\sigma)$ which is a leaf node in $T$. As argued in the base case, $\phi$ cannot be a "literal" and since $w$ does not have any successor in $\M_T$, we have $\M_T,w,\identity\models \Box\alpha_j$ for every $j\le m$. Thus, the claim holds vacuously.
\end{itemize}
\end{proof}
\begin{proofsketch}
    The proof is by reverse induction on the height of the node $\tau= (w:\Gamma,S,F)$ in $T$. If $\tau$ is a leaf node in $T$ then $\Gamma$ contains only propositions and hence by construction, we can argue that the claim holds vacuously. Otherwise, $\tau$ has one or more successors as a result of application of some rule at $\tau$. 
    
    If $\nestedForallRule$ was applied at $\tau$ then we prove that if the violation is because of a "nested $\forall$ formula" then either $\tau$ itself satisfies the required condition or there has to be a successor of $\tau$ where the violation happens and hence the claim holds inductively.

    In all other cases we argue that there is some formula that is false in the descendant and hence the claim holds by induction. The details are provided in \Cref{completeness-appendix}.
\end{proofsketch}

Now we are ready to prove the $(\Rightarrow)$ direction of the \autoref{thm-tableau for EBBE}. 

\begin{proof}[Proof of $(\Rightarrow)$ of \autoref{thm-tableau for EBBE}]
Let $T$ be the "open tableau" rooted at $(r: \{\theta\},S_0,\notInit)$. The goal is to build an "increasing domain model" $\M$ in which $\theta$ is "satisfiable". We build a sequence of "open tableau" $T_0,T_1\ldots$ is constructed by induction on $i$ such that either we halt at some $T_n$ such that the corresponding model $\M_n$ is the intended model or we take the limit of this tableau sequence and prove that $\theta$ is "satisfiable" in the model corresponding to this limiting tableau.

The sequence $T_0,T_1\ldots$ is constructed by induction on $i$. We will inductively maintain that for every $i \ge 0$ the following holds:
\begin{enumerate}
    \item \label{item-Gamma-extends} $T_{i+1}$ is an "extension" of $T_i$ "wrt" some node $\tau=(w:\Gamma,S,F)$ and some leaf variable $z$ in $F$.
    \item\label{item-S-extends} If there is a node $\tau=(w:\Gamma,S,F)$ in $T_i$ and a leaf $z\in F$ such that $\exists y~ \psi\in \forestLabel(z)$ then there is some $j> i$ such that in $T_j$ then node corresponding to $\tau$ given by $\tau' = (w:\Gamma',S',F')$ where $z$ is not a leaf node in $F'$.
\end{enumerate}

 In the base case, $T_0 = T$. \Cref{item-Gamma-extends} holds trivially.
By induction, assume that we have built $T_0,T_1\ldots T_i$. Now if $\M_i,r,\identity \models \theta$ then $\theta$ is "satisfiable" and we are done. Otherwise, from \Cref{lemma-leaf-violation} there exists some $\forall x~\psi \in \SFplus[\theta]$ such that $\forall x~\psi$ is a "nested $\forall$ formula" and there exists some descendant  node $\tau' = (w': \Gamma',S',F')$ of $\tau_i$ in $T_i$ and $z$ is a leaf in $F'$ such that $\exists y \psi' \in \forestLabel(z)$ and $\M_i,w',\identity_{[x\to z]} \not\models \exists y \psi'[x/z]$.  In this case we say $(\tau,z)$ is a pair where there is a leaf violation for $z$.

Let $\mathcal{T} = \{ (\tau,z)\mid$ in the node $\tau$ in $T_{i}$, we have a leaf violation for $z\}$. From the above argument, $\mathcal{T} \ne \emptyset$. Now we say $(\tau,z) < (\tau',z')$ if the distance from root to $z$ in the "skolem-forest" (in the tree to which $z$ belongs) of $\tau$ is less than the distance from root to $z'$ in the "skolem-forest" (in the tree to which $z'$ belongs) of $\tau'$. 
Pick some $(\tau,z)\in \mathcal{T}$ for which the above defined distance is minimal. Define  $T_{t+1}$ to be the "extension" if $T_i$ "wrt" $\tau$ for $z$. Clearly by construction, (\cref{item-Gamma-extends}) holds. 

To see that (\cref{item-S-extends}) holds, note that in the new $T_{i}$ for some $i\ge 0$ then in $T_{i+1}$ if a new leaf violation happens to a fresh variable $z$, the distance from $z$ to its corresponding root will be strictly more that the corresponding distance that was chosen in $T_i$. Hence, every leaf node violation becomes a candidate to be picked at some $j \ge i$.

\bigskip
Finally suppose $\M_i,w,\identity \not\models \theta$ for every $i\ge 0$ then note that by construction, $\W,\D$ and $\R$ remains the same for every $\M_i$. Hence let $\M_i = (\W,\D,\live_i,\R,\val_i)$. Define the limiting model $\M = (\W,\D,\live,\R,\val)$ where for every $w\in \W$ let $\live(w) = \bigcup\limits_{i\ge 0} \live_i(w)$ and for every predicate $P$ occurring in $\theta$ let $\val(w,P) = \bigcup\limits_{i\ge 0} \val_i(w,P)$. 

Note that $\M$ is well defined since every $\M_{i+1}$ extends $\M_i$ ("ie" $\D_{i} \subseteq \D_{i+1}$ and $\val_{i}(w,P) \subseteq \val_{i+1}(w,P)$). Now we claim that $\M,r,\identity \models \theta$ where $(r:\{\theta\},S_0,\notInit)$ is the root or $T_0$.

\begin{claim}
\label{claim-limit-model-works} For every $\phi\in \SFplus[\theta]$ and for every $i\ge 0$ if there is a node $\tau = (w:\Gamma,S,F)$ in $T_i$ such that $ \phi\in \Gamma$ then $\M,w,\identity \models \phi$.
\end{claim}
Clearly, from the claim it follows that $\M,r,\identity \models \theta$. Hence it remains to prove the claim.  Let $i\ge 0$ and $\tau = (w:\Gamma,S,F)$ be a node in $T_i$.  The claim is proved in \Cref{last-claim-proof}.
\end{proof}

\begin{appendixproof}[Proof of \Cref{claim-limit-model-works}]
\label{last-claim-proof}

The proof of the claim is by induction on the structure of $\phi$.

\begin{itemize}
    \item If $\phi$ is a "literal" of the form $P(x_1\ldots x_n)$ and $P(x_1\ldots x_n) \in \Gamma$ then by definition of $\val(w,P)$ we have $\M,w,\identity \models \phi$.
    \item The cases of $\land$ and $\lor$ are standard.
    \item If $\phi$ is of the form $\Diamond \psi$ and $\phi \in \Gamma$ then we consider two cases:
    \begin{itemize}
        \item[(a)] If $\psi$ is not a "nested $\forall$ formula" then there is some descendant $\tau'= (wv_i: \Gamma',S',F')$ of $\tau$ in $T_i$ where $\psi\in \Gamma'$. Hence by induction we have $\M,wv_i,\identity \models \psi$ which implies $\M,w,\identity \models \Diamond \psi$.
        \item[(b)] If $\psi$ is a "nested $\forall$ formula" of the form $\forall x~\alpha$ then there is some descendant $\tau'= (wv_i: \Gamma',S',F')$ of $\tau$ in $T_i$ which is a result of applying $\nestedForallRule$ for $\forall x~\alpha$.  Now suppose $\M,wv_i,\identity \models \forall x~\alpha$ then we are done. Otherwise let $\live(w) = \{z_0,z_1\ldots\}$ be an arbitrary ordering of $\live(w)$ and let and $z_l$ be the smallest index such that $\M,wv_i,\identity_{[x\to z_l]} \not\models \alpha$. Now since $z_l\in \live(w)$ there is some $j\ge 0$ and a node $\tau'' = (wv_i,\Gamma'',S'',F'')$ in $T_j$ such that $z_l$ is not a leaf node in $F''$. Now by construction, $\alpha\in \forestLabel''(z)$ and hence $\alpha[z_l/x] \in \Gamma''$. But then, by induction hypothesis we have $\M,wv_i,\identity \models \alpha[z_l/x]$ which implies $\M,wv_i,\identity_{[x\to z_l]}\models \alpha$ which is a contradiction. 
    \end{itemize}

    \item In case of $\Box \psi$ note that $\psi$ is not a "nested $\forall$ formula" and hence is similar to $(a)$ in the previous case.
    \item If $\phi$ is of the form $\exists y \psi$ and $\phi \in \Gamma$ then there is some descendant $\tau'= (w,\Gamma',S',F')$ of $\tau'$ in $T_i$ where $\psi\in \Gamma'$ and hence by induction we have $\M,w,\identity \models \psi$ which implies $\M,w,\identity \models \exists y \psi$.
    \item If $\phi$ is of the form $\forall y~ \psi$ then by construction, $\phi$ is not a "nested $\forall$ formula". Pick an arbitrary $z\in \live(w)$. Now since $z\in \live(w)$ there is some $j$ such that $z\in \live_j(w)$ and we also have a corresponding node $\tau' = (w:\Gamma',S',F')$ in $T_j$ where $\forall y \psi\in \Gamma'$. Hence there is some successor $\tau''= (w:\Gamma'',S'',F'')$ of $\tau'$ in $T_j$ where $ \psi[y/z] \in \Gamma''$. But then, by \Cref{lemma-leaf-violation}, $\M,w,\identity \models \psi[y/z]$. Since $z$ was picked arbitrarily, we have $\M,w,\identity \models \forall y~\psi$. \qedhere
\end{itemize}
\end{appendixproof}

\begin{corollary}
"Satisfiability problem" for $\ourEBBE$ is decidable.
\end{corollary}
\begin{proof}
    Suppose $\phi$ is satisfiable. Then from \Cref{thm-tableau for EBBE}, there is an "open tableau". From \Cref{lemma-skolem-forest-has-bounded-size}, we can always "skolem-forest" of bounded size and hence the labels at every node in the tableau is at most double exponential. Hence we directly get a "2NExpTime" algorithm.
\end{proof}

\section{Conclusion}
\label{sec-conclusion}
Despite not satisfying "finite model property", we proved that the $\ourEBBE$ bundled fragment is decidable, thus answering the open question posed in \cite{BundledJournal} and collapsing their `trichotomy' to a `dichotomy'. However note that we have only considered the "increasing domain models". There is also an other open problem which is the $\Box\exists$ fragment over "constant domain models" which is also left open in \cite{BundledJournal} for which there is no "finite model property".

 We believe that the tools and techniques developed in this paper can be used to prove decidability for $\Box\exists$ fragment over "constant domain models", but it remains to be verified. The extensions with constants, equality, frame restrictions and other such standard variations also need to be studied in details. The lower bound is also not addressed which we believe is "2ExpSpace" the problem might be "2ExpSpace"-complete. But this needs to be verified.

This paper offers a novel way to define pseudo-finite witness via "skolem-forest" which we believe can be utilized in other extensions of first order logic to prove decidability especially when the logic does not satisfy the "finite model property" (like the two variable fragment of term modal logic with equality \cite{2varTMLjournal}.

%\section*{Acknowledgments}

%ACKNOWLEDGEMENT

\appendix

%% The file kr.bst is a bibliography style file for BibTeX 0.99c
\bibliographystyle{kr}
\bibliography{long, ref}

\begin{thebibliography}{}

\bibitem[\protect\citeauthoryear{Belardinelli and Lomuscio}{2009}]{epistemic1}
Belardinelli, F., and Lomuscio, A.
\newblock 2009.
\newblock Quantified epistemic logics for reasoning about knowledge in
  multi-agent systems.
\newblock {\em Artificial Intelligence} 173(9-10):982--1013.

\bibitem[\protect\citeauthoryear{Belardinelli and Lomuscio}{2012}]{temporal1}
Belardinelli, F., and Lomuscio, A.
\newblock 2012.
\newblock Interactions between knowledge and time in a first-order logic for
  multi-agent systems: completeness results.
\newblock {\em Journal of Artificial Intelligence Research} 45:1--45.

\bibitem[\protect\citeauthoryear{B{\"o}rger, Gr{\"a}del, and
  Gurevich}{2001}]{ShelahBook}
B{\"o}rger, E.; Gr{\"a}del, E.; and Gurevich, Y.
\newblock 2001.
\newblock {\em The classical decision problem}.
\newblock Springer Science \& Business Media.

\bibitem[\protect\citeauthoryear{Dixon \bgroup et al\mbox.\egroup
  }{2008}]{temporal2}
Dixon, C.; Fisher, M.; Konev, B.; and Lisitsa, A.
\newblock 2008.
\newblock Practical first-order temporal reasoning.
\newblock In {\em 2008 15th International Symposium on Temporal Representation
  and Reasoning},  156--163.
\newblock IEEE.

\bibitem[\protect\citeauthoryear{Gabbay and Shehtman}{1993}]{FOML2varUndec1}
Gabbay, D.~M., and Shehtman, V.~B.
\newblock 1993.
\newblock Undecidability of modal and intermediate first-order logics with two
  individual variables.
\newblock {\em J. Symbolic Logic} 58(3):800--823.

\bibitem[\protect\citeauthoryear{Gradel and
  Walukiewicz}{1999}]{gradel1999guarded}
Gradel, E., and Walukiewicz, I.
\newblock 1999.
\newblock Guarded fixed point logic.
\newblock In {\em Proceedings. 14th Symposium on Logic in Computer Science
  (Cat. No. PR00158)},  45--54.
\newblock IEEE.

\bibitem[\protect\citeauthoryear{Hodkinson, Wolter, and
  Zakharyaschev}{2002a}]{temporal3}
Hodkinson, I.; Wolter, F.; and Zakharyaschev, M.
\newblock 2002a.
\newblock Decidable and undecidable fragments of first-order branching temporal
  logics.
\newblock In {\em Proceedings 17th annual IEEE symposium on logic in computer
  science},  393--402.
\newblock IEEE.

\bibitem[\protect\citeauthoryear{Hodkinson, Wolter, and
  Zakharyaschev}{2002b}]{Monodic2}
Hodkinson, I.; Wolter, F.; and Zakharyaschev, M.
\newblock 2002b.
\newblock Decidable and undecidable fragments of first-order branching temporal
  logics.
\newblock In {\em Proceedings LICS 2002},  393--402.

\bibitem[\protect\citeauthoryear{Hughes and Cresswell}{1996}]{Cresswell96}
Hughes, G.~E., and Cresswell, M.~J.
\newblock 1996.
\newblock {\em A New Introduction to Modal Logic}.
\newblock Routledge.

\bibitem[\protect\citeauthoryear{Kontchakov, Kurucz, and
  Zakharyaschev}{2005}]{FOML2varUndec2}
Kontchakov, R.; Kurucz, A.; and Zakharyaschev, M.
\newblock 2005.
\newblock Undecidability of first-order intuitionistic and modal logics with
  two variables.
\newblock {\em Bull. Symbolic Logic} 11(3):428--438.

\bibitem[\protect\citeauthoryear{Kripke}{1962}]{KripkeUndec}
Kripke, S.~A.
\newblock 1962.
\newblock The undecidability of monadic modal quantification theory.
\newblock {\em Mathematical Logic Quarterly} 8(2):113--116.

\bibitem[\protect\citeauthoryear{Li and Wang}{2019}]{planning1}
Li, Y., and Wang, Y.
\newblock 2019.
\newblock Multi-agent knowing how via multi-step plans: a dynamic epistemic
  planning based approach.
\newblock In {\em Logic, Rationality, and Interaction: 7th International
  Workshop, LORI 2019, Chongqing, China, October 18--21, 2019, Proceedings 7},
  126--139.
\newblock Springer.

\bibitem[\protect\citeauthoryear{Liu \bgroup et al\mbox.\egroup
  }{2022}]{BundledMFCS}
Liu, M.; Padmanabha, A.; Ramanujam, R.; and Wang, Y.
\newblock 2022.
\newblock Generalized bundled fragments for first-order modal logic.
\newblock In {\em 47th International symposium on mathematical foundations of
  computer science (MFCS 2022)}.
\newblock Schloss Dagstuhl-Leibniz-Zentrum f{\"u}r Informatik.

\bibitem[\protect\citeauthoryear{Liu \bgroup et al\mbox.\egroup
  }{2023}]{BundledJournal}
Liu, M.; Padmanabha, A.; Ramanujam, R.; and Wang, Y.
\newblock 2023.
\newblock Are bundles good deals for first-order modal logic?
\newblock {\em Information and Computation}  105062.

\bibitem[\protect\citeauthoryear{Padmanabha and
  Ramanujam}{2023}]{2varTMLjournal}
Padmanabha, A., and Ramanujam, R.
\newblock 2023.
\newblock A decidable fragment of first order modal logic: two variable term
  modal logic.
\newblock {\em ACM Transactions on Computational Logic} 24(4):1--38.

\bibitem[\protect\citeauthoryear{Padmanabha, Ramanujam, and
  Wang}{2018}]{BundledFSTTCS}
Padmanabha, A.; Ramanujam, R.; and Wang, Y.
\newblock 2018.
\newblock Bundled fragments of first-order modal logic: (un)decidability.
\newblock In Ganguly, S., and Pandya, P.~K., eds., {\em 38th {IARCS} Annual
  Conference on Foundations of Software Technology and Theoretical Computer
  Science, {FSTTCS} 2018, December 11-13, 2018, Ahmedabad, India}, volume 122
  of {\em LIPIcs},  43:1--43:20.
\newblock Schloss Dagstuhl - Leibniz-Zentrum f{\"{u}}r Informatik.

\bibitem[\protect\citeauthoryear{Vardi}{1997}]{vardi11997modal}
Vardi, M.~Y.
\newblock 1997.
\newblock Why is modal logic so robustly decidable?
\newblock In {\em Descriptive Complexity and Finite Models: Proceedings of a
  DIMACS Workshop, January 14-17, 1996, Princeton University}, volume~31,  149.
\newblock American Mathematical Soc.

\bibitem[\protect\citeauthoryear{Wang}{2017}]{Wang17}
Wang, Y.
\newblock 2017.
\newblock A new modal framework for epistemic logic.
\newblock In {\em Proceedings of {TARK} 2017},  515--534.

\bibitem[\protect\citeauthoryear{Wolter and Zakharyaschev}{2001}]{Monodic1}
Wolter, F., and Zakharyaschev, M.
\newblock 2001.
\newblock Decidable fragments of first-order modal logics.
\newblock {\em J. Symb. Log.} 66(3):1415--1438.

\end{thebibliography}

\end{document}